\documentclass[12pt]{amsart}
\usepackage[dvips]{color}
\usepackage{amsmath}
\usepackage{amsxtra}
\usepackage{amscd}
\usepackage{amsthm}
\usepackage{amsfonts}
\usepackage{amssymb}
\usepackage{eucal}
\usepackage{epsfig}
\usepackage{graphics}
\usepackage{ytableau}
\newcommand{\bra}[1]{\langle #1 |}        
\newcommand{\ket}[1]{{| #1 \rangle}}

\def\half{\textstyle{\frac  1 2}}

\newcommand{\C}{{\mathbb C}}
\newcommand{\Z}{{\mathbb Z}}

\newcommand{\B}{{\mathcal B}}
\newcommand{\F}{{\mathcal F}}
\newcommand{\Hc}{\mathcal{H}}
\newcommand{\V}{{\mathcal V}}
\newcommand{\N}{{\mathcal N}}

\newcommand{\id}{\textit{id}}
\newcommand{\e}{\varepsilon}
\newcommand{\Ag}{\mathfrak{a}}
\newcommand{\Bg}{\mathfrak{b}}
\newcommand{\Lc}{\mathcal{L}}
\newcommand{\vac}{\ket{\mathrm{vac}}}
\newcommand{\nn}{\nonumber}
\newtheorem{thm}{Theorem}[section]
\newtheorem{prop}[thm]{Proposition}
\newtheorem{lem}[thm]{Lemma}
\newtheorem{rem}[thm]{Remark}

\newtheorem{conj}[thm]{Conjecture}

\numberwithin{equation}{section}

\newcommand{\bla}{\boldsymbol{\lambda}}
\newcommand{\bmu}{\boldsymbol{\mu}}

\begin{document}
\begin{title}[CFT approach to $q$-PVI]
{CFT approach to the $q$-Painlev{\'e} VI equation}
\end{title}
\author{M. Jimbo, H. Nagoya, and H. Sakai}
\address{MJ: Department of Mathematics,
Rikkyo University, Toshima-ku, Tokyo 171-8501, Japan}
\email{jimbomm@rikkyo.ac.jp}
\address{HN: School of Mathematics and Physics, Kanazawa University, Kanazawa, Ishikawa 920-1192, Japan}
\email{nagoya@se.kanazawa-u.ac.jp}
\address{HS: Graduate School of Mathematical Sciences,
University of Tokyo, Meguro-ku, Tokyo 153-8914, Japan}
\email{sakai@ms.u-tokyo.ac.jp}

\begin{abstract} 
Iorgov, Lisovyy, and Teschner established a 
connection between isomonodromic deformation of linear differential equations 
and Liouville conformal field theory at $c=1$. 
In this paper we present a $q$ analog of their construction. 
We show that the general solution of the 
 $q$-Painlev{\'e} VI equation is a ratio of four tau functions, 
each of which is 
given by a combinatorial series arising in the AGT correspondence. 
We also propose conjectural bilinear equations for the tau functions.
\end{abstract}

\dedicatory{
Dedicated to Professor Kazuo Okamoto on the occasion of
his seventieth birthday}
\maketitle

\section{Introduction}

Theory of isomonodromic deformation and Painlev{\'e} equations
has long been a subject of intensive study, with many applications both in 
mathematics and in physics.  
Recent discovery of the isomonodromy/CFT correspondence 
\cite{GIL}, \cite{GIL1}, \cite{ILTy}
opened up a new perspective to this domain. 
In a nutshell, it states that the Painlev{\'e}
tau functions are Fourier transform of conformal blocks at $c=1$. 
The AGT relation  allows one to write a combinatorial 
formula for the latter, thereby giving an explicit series representation
for generic Painlev{\'e} tau functions. 

The initial conjectures of 
\cite{GIL}, \cite{GIL1}, \cite{ILTy}
concern the sixth Painlev\'e (PVI) equation, and at present 
three different proofs are known. 
In the first approach \cite{ILTe} 
one constructs solutions to the linear problem 
out of conformal blocks with the insertion of a degenerate primary field.
The second approach \cite{BS} also makes use of CFT ideas, 
but leads directly to bilinear differential equations for tau functions. 
The third approach \cite{GL} is based on the Riemann-Hilbert problem. 
It is shown that Fredholm determinant representation
for tau functions reproduces the combinatorial series expansion,  
without recourse to conformal blocks or the AGT relation. 
Painlev{\'e} equations of other types are also studied 
by confluence \cite{GIL1},  
in the framework of irregular conformal blocks 
\cite{N, BS2},
through the AGT correspondence \cite{BLMST}, and using Fredholm determinants \cite{GL1}. 

It is natural to ask whether a similar picture holds  for 
the $q$ difference analog of Painlev{\'e} equations. 
This has been studied in \cite{BS1} 
for the qPIII($D_8$) equation. 
The aim of the present paper is to give
a $q$ analog of the construction of \cite{ILTe}.
 
Our method is based on the work \cite{AFS}, 
where the five-dimensional AGT correspondence 
was studied in the light of the quantum toroidal $\mathfrak{gl}_1$ algebra
(also known as the Ding-Iohara-Miki algebra).  
Having the central charge $c=1$ corresponds to choosing 
$t=q$ where $t,q$ are the parameters of the algebra.  
This is a sort of classical limit, 
in the sense that the quantum algebra reduces to 
the enveloping algebra of a Lie algebra, see Section 2.1.

We use the trivalent vertex operators  of \cite{AFS} to
introduce a $q$ analog of primary fields and conformal blocks. 
As in the continuous case, a key property is the braiding relation 
in the presence of the degenerate field. 
Following the scheme of \cite{ILTe}, 
we then write the solution $Y(x)$ of the Riemann problem 
as Fourier transform of appropriate conformal block functions.
For the sake of concreteness, 
we restrict our discussion to the setting 
relevant to the qPVI equation. 
Also we do not discuss the issue of convergence of these series.  

The above construction allows us to relate  
the unknown functions $y,z$ of qPVI to tau functions. 
More specifically, we express $y,z$ in terms of four tau functions
which arise from the expansions at $x=0$ or $x=\infty$; see
\eqref{y-tau}, \eqref{z-tau} below. 
On the other hand, 
in view of the work \cite{S}, \cite{TM} we expect that there are 
four more tau functions which naturally 
enter the picture; these are the ones related to 
the singular points other than $x=0,\infty$. 
However, for the lack of the notion of fusion, 
it is not clear to us how to related them to $y,z$. 
The absence of fusion is the main obstacle in the $q$ analog of 
isomonodromy theory. 
We present conjectural bilinear relations
satisfied by these eight tau functions, see \eqref{bilin-1}--\eqref{bilin-8}.  
Assuming the conjecture we give the final expression for $y,z$
in \eqref{yz-final}.
\bigskip

The plan of the paper is as follows.

In Section 2, we give a brief account of the results of \cite{AFS}, 
restricting to the special case $t=q$. 
We introduce a $q$ analog of primary fields, 
conformal blocks, degenerate $(2,1)$ fields and their braiding. 
In Section 3, we apply them to the $q$ analog of the Riemann problem, 
focusing attention to the setting of the $q$PVI equation. 
We derive a formula 
expressing the unknowns $y,z$ in terms of four tau functions. 
We then propose conjectural bilinear difference
equations for eight tau functions. 

In Appendix we give a direct combinatorial proof of the braiding relation
used in the text.
\bigskip

\textit{Notation.}\quad Throughout the paper we fix $q\in\C^\times$ such that $|q|<1$. 
We set 
$[u]=(1-q^u)/(1-q)$,  
$(a;q)_N=\prod_{j=0}^{N-1}(1-aq^j)$, 
$(a_1,\ldots,a_k;q)_{\infty}=\prod_{j=1}^k(a_j;q)_\infty$, and 
$(a;q,q)_\infty=\prod_{j,k=0}^\infty(1-aq^{j+k})$.
We use the $q$ Gamma function, $q$ Barnes function and the theta function 
\begin{align*}
&\Gamma_q(u)=\frac{(q;q)_\infty}{(q^u;q)_\infty}(1-q)^{1-u},
\quad 
G_q(u)=\frac{(q^u;q,q)_\infty}{(q;q,q)_\infty}(q;q)_\infty^{u-1}(1-q)^{-(u-1)(u-2)/2},
\\
&\vartheta(u)=q^{u(u-1)/2}\Theta_q(q^u),\quad 
\Theta_q(x)=(x,q/x,q;q)_\infty, 
\end{align*}
which satisfy $\Gamma_q(1)=G_q(1)=1$ and 
\begin{align*}
&\Gamma_q(u+1)=[u]\Gamma_q(u),  \quad 
G_q(u+1)=\Gamma_q(u)G_q(u), \\
&\vartheta(u+1)=-\vartheta(u)=\vartheta(-u)\,.
\end{align*}

A partition is a finite sequence of positive integers 
$\lambda=(\lambda_1,\ldots,\lambda_l)$ such that 
$\lambda_1\ge\ldots\ge\lambda_{l}>0$. We set $\ell(\lambda)=l$. 
The conjugate partition $\lambda'=(\lambda'_1,\ldots,\lambda'_{l'})$
is defined by $\lambda'_j=\sharp\{i\mid \lambda_i\ge j\}$,
$l'=\lambda_1$.   
We denote by $\Lambda$ the set of all partitions. 
We regard a partition $\lambda$ also as the subset 
$\{(i,j)\in\Z^2\mid 1\le j\le \lambda_i,\ i\ge 1\}$ of $\Z^2$,  
and denote its cardinality by $|\lambda|$. 
For $\square=(i,j)\in\Z_{>0}^2$ we set $a_\lambda(\square)=\lambda_i-j$
and $\ell_\lambda(\square)=\lambda'_j-i$.
In the last formulas we set $\lambda_i=0$ if $i>\ell(\lambda)$
(resp. $\lambda'_j=0$ if $j>\ell(\lambda')$).

\section{$q$ analog of conformal blocks}

In this section, we collect background materials 
from \cite{AFS}. 
We shall restrict ourselves to the special case $t=q$, 
where formulas of \cite{AFS} simplify considerably. 
We then introduce a $q$ analog of primary fields, 
conformal blocks and the braiding relation for the degenerate fields 
(analog of $(2,1)$ operators).

\subsection{Toroidal Lie algebra}
Let $Z,D$ be non-commuting variables satisfying $ZD=qDZ$. 
The ring of Laurent polynomials $\C\langle Z^{\pm1},D^{\pm 1}\rangle$ 
in $Z,D$
is a Lie algebra, with the Lie bracket given by the commutator.
We consider its two-dimensional central extension, which we denote by $\Lc$.   
As a vector space $\Lc$ has $\C$-basis $Z^kD^l$ ($(k,l)\in\Z^2\setminus\{(0,0)\}$)
and central elements $c_1,c_2$.
The Lie bracket is defined by 
\begin{align*}
&[Z^{k_1}D^{l_1},Z^{k_2}D^{l_2}]=(q^{-l_1k_2}-q^{-l_2k_1})Z^{k_1+k_2}D^{l_1+l_2}
+\delta_{k_1+k_2,0}\delta_{l_1+l_2,0}q^{-k_2l_1} (k_1c_1+l_1 c_2)\,.
\end{align*}
We set 
\begin{align*}
&a_r=Z^r,\quad x^+_n=DZ^n,\quad  x^-_n=Z^nD^{-1}\,,
\\
& b_r=D^{-r}, \quad y^+_n=Z D^{-n},\quad  y^-_n=D^{-n}Z^{-1}\,,
\end{align*}
and introduce the generating series 
$x^\pm(z)=\sum_{n\in\Z} x^\pm_nz^{-n}$,  $y^\pm(z)=\sum_{n\in\Z} y^\pm_nz^{-n}$. 
We have two Heisenberg subalgebras of $\Lc$,  $\Ag=\oplus_{r\neq 0}\C a_r\oplus\C c_1$ and 
$\Bg=\oplus_{r\neq 0}\C b_r\oplus\C c_2$: 
$[a_r,a_s]=r\delta_{r+s,0}c_1$, $[b_r,b_s]=-r\delta_{r+s,0}c_2$. 
We say that an $\Lc$-module has level $(l_1,l_2)$ if 
$c_1$ acts as $l_1\cdot \id$ and $c_2$ acts as $-l_2\cdot \id$.  

The Lie algebra $\Lc$ has automorphisms $S,T$ given by 
\begin{align*}
&S:Z\mapsto D,\quad D\mapsto Z^{-1},\quad c_1\mapsto c_2,
\quad c_2\mapsto -c_1,\\ 
&T:Z\mapsto Z,\quad D\mapsto DZ,\quad c_1\mapsto c_1,\quad c_2\mapsto c_1+c_2\,.
\end{align*}
They satisfy $S^4=(ST)^6=\id$, so that the group $SL(2,\Z)$ 
acts on $\Lc$ by automorphisms. We have $S(b_r)=a_r$, 
$S(y^{\pm}(z))=x^\pm(z)$.

\subsection{Fock representations}
The most basic representations of $\Lc$
are the Fock representations of levels $(1,0)$ and $(0,1)$. 
Following \cite{AFS} we denote them by 
$\F^{(1,0)}_u$ and $\F^{(0,1)}_u$ ($u\in \C^\times$), respectively. 

The Fock representation $\F^{(1,0)}_u$ is 
irreducible under the Heisenberg subalgebra $\Ag$. 
As a vector space we have $\F^{(1,0)}_u=\C[a_{-1},a_{-2},\ldots]\vac$, where 
$\vac$ is a cyclic vector such that $a_r\vac=0$  ($r>0$).
The action of the generators $x^\pm_n$ is given by vertex operators
\begin{align*}
&x^\pm(z)\mapsto \frac{1}{1-q^{\mp1}}u^{\pm1}
\exp\bigl(\mp\sum_{r\ge1}\frac{1-q^r}{r}a_{-r}z^r\bigr) 
\exp\bigl(\pm\sum_{r\ge1}\frac{1-q^{-r}}{r}a_{r}z^{-r}\bigr)\,.
\end{align*}

The Fock representation $\F^{(0,1)}_u$ is the pullback 
of $\F^{(1,0)}_u$ by the automorphism $S$. 
It is irreducible with respect to the Heisenberg subalgebra $\Bg$, and 
$y^\pm(z)$ act as vertex operators. 
On the other hand, the action of $\Ag$ on $\F^{(0,1)}_u$ is commutative.
The joint eigenvectors 
$\{\ket{\lambda}\}_{\lambda\in\Lambda}$ of $\Ag$
are labeled by all partitions and constitute a basis of $\F^{(0,1)}_u$.  
In particular, for the vector $\ket{\emptyset}$ we have
\begin{align*}
a_r\ket{\emptyset}=-\frac{u^r}{1-q^r}\ket{\emptyset}\,,
\quad x^-(z) \ket{\emptyset}=0,
\quad x^+(z) \ket{\emptyset}=\delta(u/z)\ket{(1)}
\end{align*}
where $\delta(x)=\sum_{n\in\Z}x^n$. 
Formulas for the action of $a_r$ and $x^{\pm}(z)$ 
on a general $\ket{\lambda}$ 
can also be written explicitly. Since we do not use them in this paper, 
we refer the reader to \cite{FT,AFS}. 

Tensor products of $N$ copies of $\F^{(1,0)}_u$ are closely related to 
representations of the $W_N$ algebra (see e.g. \cite{FHHSY}, \cite{Mi}).
We briefly explain this point in the case $N=2$. 

Consider the tensor product of two Fock representations and its decomposition 
with respect to the Heisenberg subalgebra $\Ag$,  
\begin{align*}
\F^{(1,0)}_{u_1}\otimes \F^{(1,0)}_{u_2}=\Hc\otimes \Omega_{u_1,u_2}\,,
\end{align*}
where $\Hc$ is the Fock space of $\Ag$ 
and $\Omega_{u_1,u_2}$ is the multiplicity space. 
The operators $x^{\pm}(z)$ act on $\F^{(1,0)}_{u_1}\otimes \F^{(1,0)}_{u_2}$ 
as a sum of two vertex operators.
We can factor it into a product of two commuting operators, 
\begin{align*}
(1-q^{\mp1})x^{\pm}(z)\mapsto g^\pm(z)\cdot (u_1u_2)^{\pm1/2}T(z;(u_1/u_2)^{\pm1/2}). 
\end{align*}
The first factor represents the action of $\Ag$ on $\Hc$, 
\begin{align*}
g^\pm(z)=:\exp\left(\pm\sum_{r\neq0}\frac{1-q^{-r}}{2r}
\bigl(a^{(1)}_r+a^{(2)}_r\bigr)z^{-r}\right):\,,
\end{align*}
where $a_r^{(1)}=a_r\otimes \id$, $a_r^{(2)}=\id\otimes a_r$. 
The second factor acts on $\Omega_{u_1,u_2}$, 
\begin{align*}
&T(z;u)=u \Lambda^+(z)+u^{-1}\Lambda^-(z),\\ 
&\Lambda^{\pm}(z)=:\exp\left(\pm\sum_{r\neq0}\frac{1-q^{-r}}{2r}
\bigl(a^{(1)}_r-a^{(2)}_r\bigr)z^{-r}\right):
\end{align*}
and gives 
a free field realization of the deformed Virasoro algebra \cite{SKAO}. 
The space $\Omega_{u_1,u_2}$ is a $q$ analog of 
the Virasoro Verma module with central charge $c=1$
and highest weight $\Delta_\theta=\theta^2$, 
where $q^{2\theta}=u_1/u_2$. 

\subsection{Trivalent vertex operators}
For $N\in\Z$,  we denote by $\F^{(1,N)}_w$ the pullback of $\F^{(1,0)}_w$
by the automorphism $T^N$\kern-2pt.
The following intertwiners of $\Lc$-modules are called 
trivalent vertex operators. 
\begin{align*}
&\Phi=\Phi\Bigl[{{(1,N+1),wx}\atop {(0,1),-w;(1,N),x}}\Bigr]
:\F^{(0,1)}_{-w}\otimes \F^{(1,N)}_x \longrightarrow \F^{(1,N+1)}_{wx}\,,
\\ 
&\Phi^*=\Phi^*\Bigl[{{(1,N),x;(0,1),-w}\atop {(1,N+1),wx}}\Bigr]
:\F^{(1,N+1)}_{wx}\longrightarrow
\F^{(1,N)}_x \otimes  \F^{(0,1)}_{-w}\,.
\end{align*}
Such non-trivial intertwiners exist  
and are unique up to a scalar multiple \cite{AFS}. 
We introduce their components 
$\widehat{\Phi}_\lambda(w),\widehat{\Phi}^*_\lambda(w)$ 
relative to the basis $\{\ket{\lambda}\}$ of $\F^{(0,1)}_{-w}$ as follows. 
\begin{align*}
&\Phi\bigl(\ket{\lambda}\otimes \alpha\bigr)= 
\hat{t}(\lambda,x,w,N)
\cdot 
\widehat{\Phi}_\lambda(w)(\alpha)\quad (\alpha\in \F^{(1,N)}_x)\,,
\\
&\Phi^* (\alpha) =\sum_{\lambda\in\Lambda} 
\hat{t}^*(\lambda,x,w,N)
\cdot \widehat{\Phi}^*_\lambda(w)(\alpha)\otimes \ket{\lambda}
\quad (\alpha\in \F^{(1,N+1)}_{wx})
\,.
\end{align*}
The numerical factors in front are 
\begin{align*}
\hat{t}(\lambda,x,w,N)
=\frac{q^{n(\lambda')}}{c_\lambda}f_\lambda^{-N-1}
\left(\frac{x}{w^N}\right)^{|\lambda|}\,,
\quad 
\hat{t}^*(\lambda,x,w,N)
=\frac{q^{n(\lambda')}}{c_\lambda}f_\lambda^N
\left(\frac{q w^N}{x}\right)^{|\lambda|}
\,,
\end{align*}
where
\begin{align*}
&n(\lambda) =\sum_{\square\in\lambda}\ell_\lambda(\square),\quad 
n(\lambda') =\sum_{\square\in\lambda}a_\lambda(\square), \\
&f_\lambda=(-1)^{|\lambda|}q^{n(\lambda')-n(\lambda)}\,,
\quad
c_\lambda=\prod_{\square\in\lambda}(1-q^{\ell_\lambda(\square)+a_\lambda(\square)+1})\,. 
\end{align*}
We normalize $\widehat{\Phi}(w)$, $\widehat{\Phi}^*(w)$ so that 
 $\widehat{\Phi}_\emptyset(w)\vac=\vac+\cdots$, 
 $\widehat{\Phi}^*_\emptyset(w)\vac=\vac+\cdots$. 

\begin{prop}\cite{AFS}
The operators $\widehat{\Phi}_\lambda(w),\widehat{\Phi}^*_\lambda(w)$
are given by 
\begin{align*}
&\widehat{\Phi}_\lambda(w)=:\widehat{\Phi}^+_\emptyset(w) \eta^+_\lambda(w):,
\quad
\widehat{\Phi}^*_\lambda(w)=:\widehat{\Phi}^-_\emptyset(w) \eta^-_\lambda(w):,
\\
&\eta^\pm_\lambda(w)=:\prod_{i=1}^{\ell(\lambda)}\prod_{j=1}^{\lambda_i} 
\eta^\pm(q^{j-i}):\,,
\end{align*}
where
\begin{align*}
&\widehat{\Phi}^\pm_\emptyset(w)=
:\exp\left(\mp\sum_{n\neq 0}\frac{1}{n}\frac{q^n}{1-q^n}a_{n}w^{-n}\right):
\,,
\\
&\eta^\pm(w)=:\exp\left(\mp\sum_{n\neq0}\frac{1-q^{n}}{n}a_n w^{-n}\right):\,.
\end{align*}
\qed
\end{prop}

Of particular importance is their normal ordering rule. 
For a pair of partitions $(\lambda,\mu)\in\Lambda^2$, 
we introduce the Nekrasov factor by 
\begin{align*}
&N_{\lambda,\mu}(w)
=\prod_{\square\in\lambda}(1-q^{-\ell_\lambda(\square)-a_{\mu}(\square)-1}w) 
\prod_{\square\in\mu}(1-q^{a_{\lambda}(\square)+\ell_\mu(\square)+1}w) \,.
\end{align*}

\begin{prop}\cite{AFS}
Let $X,Y$ be either $\widehat{\Phi}$ or  $\widehat{\Phi}^*$. 
Then the following normal ordering rule holds. 
\begin{align}
&X_\lambda(z) Y_\mu(w)=
 :X_\lambda(z) Y_\mu(w):\times
\begin{cases}
\bigl(\widehat{G}_q(q w/z) 
N_{\mu,\lambda}(w/z)\bigr)^{-1}
 & \text{for $X=Y$},\\
\widehat{G}_q(q w/z) N_{\mu,\lambda}(w/z)
 & \text{for $X\neq Y$},\\
\end{cases}
\label{normal}
\end{align} 
where $\widehat{G}_q(z)=(z;q,q)_\infty$. 
\qed
\end{prop}
It is useful to note the relations
\begin{align}
&N_{\lambda,\mu}(w)=N_{\mu,\lambda}(w^{-1}) 
w^{|\lambda|+|\mu|}\frac{f_\lambda}{f_\mu},
\label{rule1}\\
&\left(\frac{q^{n(\lambda')}}{c_\lambda}\right)^2=\frac{f_\lambda q^{-|\lambda|}}{N_{\lambda,\lambda}(1)}\,.
\label{rule2}
\end{align}

\subsection{Chiral primary field}
In what follows we set 
\begin{align*}
\V_{\theta,w}:=\F^{(0,1)}_{-q^{-\theta}w}\otimes \F^{(0,1)}_{-q^{\theta}w}\,. 
\end{align*}
We fix its basis 
$\ket{\bla}=\ket{\lambda_+}\otimes\ket{\lambda_-}$
labeled by a pair $\bla=(\lambda_+,\lambda_-)$ of partitions. 

Consider the compositions of trivalent vertex operators 
\begin{align*}
&\Psi_1:\F^{(1,0)}_{q^{-2\theta_2}/x}
\longrightarrow 
\F^{(1,-1)}_{q^{-\theta_3-\theta_2}/(xw)}\otimes 
\F^{(0,1)}_{-q^{\theta_3-\theta_2}w}
\longrightarrow 
\F^{(1,-2)}_{1/(xw^2)}\otimes \F^{(0,1)}_{-q^{-\theta_3-\theta_2}w}
\otimes \F^{(0,1)}_{-q^{\theta_3-\theta_2}w}\,,
\\
&\Psi_2:
\F^{(0,1)}_{-q^{-\theta_1}w}\otimes \F^{(0,1)}_{-q^{\theta_1}w}\otimes \F^{(1,-2)}_{1/(xw^2)}
\longrightarrow
\F^{(0,1)}_{-q^{-\theta_1}w}\otimes \F^{(1,-1)}_{q^{\theta_1}/(xw)}
\longrightarrow 
\F^{(1,0)}_{1/x}\,.
\end{align*}
By composing them we obtain an intertwiner of $\Lc$-modules
\begin{align}
\V_{\theta_1,w}\otimes \F^{(1,0)}_{q^{-2\theta_2}/x}
\overset{\id\otimes\Psi_1}{\longrightarrow}
\V_{\theta_1,w}\otimes \F^{(1,-2)}_{1/(xw^2)}
\otimes \V_{\theta_3,q^{-\theta_2}w}
\overset{\Psi_2\otimes\id}{\longrightarrow}
\F^{(1,0)}_{1/x}\otimes\V_{\theta_3,q^{-\theta_2}w}\,.
\label{comp}
\end{align}
Taking further the vacuum-to-vacuum matrix coefficient of \eqref{comp}, 
we obtain a map 
\begin{align}
&V\Bigl({{\phantom{h_1} \theta_2 \phantom{h_3}}
\atop{\theta_3 \phantom{h_2} \theta_1}};w,x\Bigr)
:\V_{\theta_1,w}\longrightarrow \V_{\theta_3,q^{-\theta_2}w}\,.
\label{primary}
\end{align}
The matrix coefficients of \eqref{primary} are given by 
\begin{align}
&\langle{\bmu}\Bigl|
V\Bigl({{\phantom{h_1} \theta_2 \phantom{h_3}}
\atop{\theta_3 \phantom{h_2} \theta_1}};w,x\Bigr)
\Bigr|\bla\rangle 
\label{V-matelt}\\
&=\mathcal{C}\, 
\frac{q^{n(\lambda'_+)}}{c_{\lambda_+}}\frac{q^{n(\lambda'_-)}}{c_{\lambda_-}}
\frac{q^{n(\mu'_+)}}{c_{\mu_+}}\frac{q^{n(\mu'_-)}}{c_{\mu_-}}
\frac{f_{\lambda_-}}{f_{\mu_+}^2f_{\mu_-}}q^{2\theta_1|\lambda_-|
+(2(\theta_3+\theta_2)+1)|\mu_+|+(2\theta_2+1)|\mu_-|}\,
\nn\\
&\times 
x^{|\bmu|-|\bla|}\cdot
\bra{\mathrm{vac}}
\widehat{\Phi}_{\lambda_+}(q^{-\theta_1}w)
\widehat{\Phi}_{\lambda_-}(q^{\theta_1}w)
\widehat{\Phi}^*_{\mu_+}(q^{-\theta_3-\theta_2}w)
\widehat{\Phi}^*_{\mu_-}(q^{\theta_3-\theta_2}w)\ket{\mathrm{vac}}\,,
\nn
\end{align}
where we set $|\bla|=|\lambda_+|+|\lambda_-|$ for 
$\bla=(\lambda_+,\lambda_-)\in\Lambda^2$, and 
$\mathcal{C}$ is a scalar factor
independent of $\bla,\bmu\in\Lambda^2$. 

We choose the normalization 
\begin{align}
&\langle  
V\Bigl({{\phantom{h_3} \theta_2 \phantom{h_1}}
\atop{\theta_3 \phantom{h_2} \theta_1}};w,x\Bigr)
\rangle
=
\N
\Bigl({{\phantom{h_3} \theta_2 \phantom{h_1}}
\atop{\theta_3 \phantom{h_2} \theta_1}}\Bigr)
q^{2\theta_2\theta_3^2}x^{\theta_3^2-\theta_2^2-\theta_1^2}\,,
\nn
\\
&\N
\Bigl({{\phantom{h_3} \theta_2 \phantom{h_1}}
\atop{\theta_3 \phantom{h_2} \theta_1}}\Bigr)
=
\frac{\prod_{\epsilon,\epsilon'=\pm}
G_q(1+\epsilon\theta_3-\theta_2-\epsilon'\theta_1)}
{G_q(1+2\theta_3)G_q(1-2\theta_1)}\,.
\label{Nfac}
\end{align}
Here and after, 
$\langle\cdots\rangle$ stands for the matrix coefficient 
$\langle(\emptyset,\emptyset)|\cdots|(\emptyset,\emptyset\rangle)$.  

We regard \eqref{primary}
as a $q$ analog of the chiral primary field of conformal dimension 
$\Delta_{\theta_2}=\theta_2^2$.

\subsection{Conformal block function}

The $(m+2)$ point conformal block function is the expectation value of a product of 
chiral primary fields, 
\begin{align*}
&\F\Bigl(
{
{\phantom{000} \theta_m\phantom{h000}
\theta_{m-1}\phantom{h0}\phantom{h0}
\cdots\phantom{h0}
\phantom{h0}\theta_1\phantom{h0}}
\atop
{\theta_{m+1}\phantom{h0}\sigma_{m-1}
\phantom{h0}\sigma_{m-2}\phantom{h0}
\cdots
\phantom{h_0}
\sigma_{1}\phantom{h0}\theta_{0}}
}
;x_m,\ldots,x_{1}
\Bigr) 
\\
&=\langle
V\Bigl({\kern-1pt\theta_m
\atop{\theta_{m+1} \phantom{h_2} \sigma_{m-1}}};w_m,x_m\Bigr)
V\Bigl({{\theta_{m-1}\phantom{h}}
\atop{\sigma_{m-1} \phantom{h_2} \sigma_{m-2}}};w_{m-1},x_{m-1}\Bigr)
\cdots
V\Bigl({{\theta_1 \phantom{h}}
\atop{\sigma_1 \phantom{h} \theta_{0}}};w_{1},x_{1}\Bigr)
\rangle
\end{align*}
where $w_p=q^{-\theta_{p-1}-\cdots-\theta_1}w_1$. 
The right hand side is actually independent of $w_1$. 

Using formula \eqref{V-matelt}, applying the normal ordering rule \eqref{normal} and further using
\eqref{rule1}, \eqref{rule2} for simplification, 
we obtain the following explicit expression. 
\begin{align}
&\F\Bigl(
{
{\phantom{000} \theta_m\phantom{h000}
\theta_{m-1}\phantom{h0}\phantom{h0}
\cdots\phantom{h0}
\phantom{h0}\theta_1\phantom{h0}}
\atop
{\theta_{m+1}\phantom{h0}\sigma_{m-1}
\phantom{h0}\sigma_{m-2}\phantom{h0}
\cdots
\phantom{h_0}
\sigma_{1}\phantom{h0}\theta_{0}}
}
;x_m,\ldots,x_{1}
\Bigr) 
\label{block}\\
&=
\prod_{p=1}^{m}\N\Bigl({\kern-1pt\theta_p\atop{\sigma_{p} \phantom{h_2} \sigma_{p-1}}}
\Bigr)q^{2\theta_p\sigma_{p}^2}
\cdot 
\prod_{p=1}^{m}
x_p^{\sigma_{p}^2-\theta_p^2-\sigma_{p-1}^2}
\nn\\
&\times \sum_{\bla^{(1)},\ldots,\bla^{(m-1)}}
\prod_{p=1}^{m-1}\Bigl(
\frac{q^{2\theta_{p}}x_p}{x_{p+1}}
\Bigr)^{|\bla^{(p)}|}\cdot
\frac{\prod_{p=1}^m \prod_{\epsilon,\epsilon'=\pm}
N_{\lambda^{(p)}_\epsilon,\lambda^{(p-1)}_{\epsilon'}}(q^{\epsilon \sigma_{p}-\theta_p-\epsilon'\sigma_{p-1}})}
{\prod_{p=1}^{m-1}\prod_{\epsilon,\epsilon'=\pm}
N_{\lambda^{(p)}_{\epsilon},\lambda^{(p)}_{\epsilon'}}(q^{\epsilon\sigma_p-\epsilon'\sigma_p})
}\,.
\nn
\end{align}
We have set
$\sigma_0=\theta_0$, $\sigma_m=\theta_{m+1}$, 
$\bla^{(0)}=\bla^{(m)}=(\emptyset,\emptyset)$,
and 
the sum is taken over all 
$\bla^{(p)}=(\lambda^{(p)}_+,\lambda^{(p)}_-)\in\Lambda^2$, $p=1,\ldots,m-1$. 

Formula \eqref{block} is the $q$ version of the AGT formula 
for conformal blocks \cite{AY} specialized to the case $t=q$. 
The derivation given above 
from the quantum algebra viewpoint is due to \cite{AFS}.

\begin{figure}
\setlength{\unitlength}{1mm}
\begin{picture}(120,60)(20,0)
\put(20,11){$\theta_{m+1}$}
\put(135,11){$\theta_{0}$}
\put(42,7){$(x_m)$}
\put(62,7){$(x_{m-1})$}
\put(77,7){$\cdots$}
\put(92,7){$(x_{2})$}
\put(112,7){$(x_1)$}
\put(52,15){$\sigma_{m-1}$}
\put(72,15){$\sigma_{m-2}$ \quad $\cdots$}
\put(102,15){$\sigma_{1}$}
\put(42,38){$\theta_m$}
\put(62,38){$\theta_{m-1}$}
\put(92,38){$\theta_{2}$}
\put(112,38){$\theta_1$}
\put(45,12){\thicklines\line(0,1){20}}
\put(45,27){\thicklines\vector(0,-1){5}}
\put(65,12){\thicklines\line(0,1){20}}
\put(65,27){\thicklines\vector(0,-1){5}}
\put(95,12){\thicklines\line(0,1){20}}
\put(95,27){\thicklines\vector(0,-1){5}}
\put(115,12){\thicklines\line(0,1){20}}
\put(115,27){\thicklines\vector(0,-1){5}}
\put(42,12){\thicklines\vector(-1,0){5}}
\put(58,12){\thicklines\vector(-1,0){5}}
\put(88,12){\thicklines\vector(-1,0){5}}
\put(108,12){\thicklines\vector(-1,0){5}}
\put(127,12){\thicklines\vector(-1,0){5}}
\put(30,12){\thicklines\line(1,0){100}}
\end{picture}

\caption{$q$-conformal block function}
\end{figure}
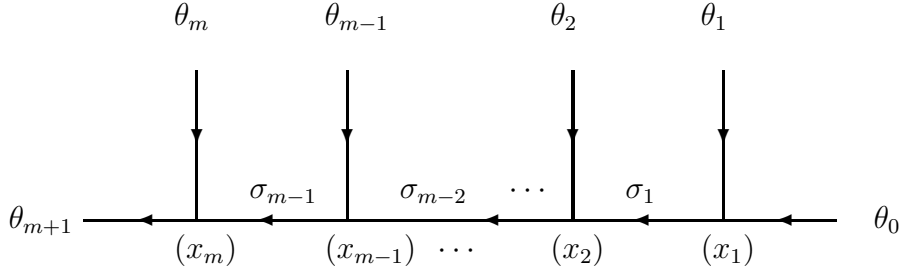

\subsection{Degenerate field and braiding}
The normalization factor \eqref{Nfac}
has a zero 
when 
$\theta_2=1/2$ 
and 
$\theta_1-\theta_3=\pm 1/2$. 
In this case we redefine the primary field by 
\begin{align*}
\N'\Bigl({{\! \half \phantom{h_3}}
\atop{\theta \phantom{h} \ \ \theta\pm \half}}\Bigr)
=
\lim_{\e\to 0}\,\frac{1}{G_q(\e)}\,
\N\Bigl({{\phantom{h} \half-\e \phantom{hhh}}
\atop{\theta \phantom{h_2} \ \ \theta\pm \half}}\Bigr)\,,
\\
V'\Bigl({{\phantom{h} \half \phantom{hhh}}
\atop{\theta \phantom{h_2} \theta\pm \half}};w,x\Bigr)
=
\lim_{\e\to 0}\,\frac{1}{G_q(\e)}\,
V\Bigl({{\phantom{h} \half-\e \phantom{hhh}}
\atop{\theta \phantom{h_2} \ \ \theta\pm \half}};w,x\Bigr)\,. 
\end{align*}

We consider the conformal block function when one of the primary fields is 
$V'\Bigl(\cdots\Bigr)$. 
We keep using the same letter $\F\left(\cdots\right)$ to denote them. 
By choosing $m=2$ and 
specializing parameters in \eqref{block}, the
four point conformal blocks can be evaluated in terms of Heine's basic 
hypergeometric series 
\begin{align*}
F\left({{\alpha,\beta}\atop{\gamma}};x\right) 
=\frac{\Gamma_q(\alpha)\Gamma_q(\beta)}{\Gamma_q(\gamma)}
{}_2\kern-1pt\phi_1\left({{q^\alpha,q^\beta}\atop{q^\gamma}};q,x\right)\, 
\end{align*}
as follows. 
\begin{align*}
&\F\Bigl({{\kern-3pt \half\phantom{hhhhh} \theta_1\phantom{hh}}
\atop
{\theta_\infty\phantom{h0}\theta_\infty+\textstyle{\frac{\e'}{2}}
\phantom{h0}\theta_0}}
;x_1,x_2\Bigr) 
= \mathcal{N}\,
q^{\theta_\infty^2}\left(\frac{q^{2\theta_1}x_2}{x_1}\right)^{\e'\theta_\infty+1/4}
\\
&\qquad \times 
F\left({{\half+\e'\theta_\infty-\theta_1+\theta_0,
\half+\e'\theta_\infty-\theta_1-\theta_0}
\atop{1+2\e'\theta_\infty}};\frac{q^{2\theta_1}x_2}{x_1}\right) \,,
\\
&\F\Bigl({{ \theta_1\phantom{hhh} \half \phantom{h}}
\atop
{\theta_\infty\phantom{h0}\theta_0+\textstyle{\frac{\e}{2}}
\phantom{h0}\theta_0}}
;x_2,x_1\Bigr) 
=\mathcal{N}\, q^{\theta_0^2}\left(\frac{qx_1}{x_2}\right)^{\e\theta_0+1/4}
\\
&\qquad\times 
F\left({{\half+\theta_\infty-\theta_1+\e\theta_0,
\half-\theta_\infty-\theta_1+\e\theta_0}
\atop{1+2\e\theta_0}};\frac{qx_1}{x_2}\right) \,.
\end{align*}
Here we have set 
\begin{align*}
\mathcal{N}=q^{2\theta_1\theta_\infty^2} 
x_1^{-1/4}x_2^{\theta_\infty^2-\theta_1^2-\theta_0^2}
\frac{\prod_{\mu,\mu'=\pm}
G_q\Bigl(\frac{1}{2}+\mu\theta_\infty-\theta_1+\mu'\theta_0\Bigr)}
{G_q(1+2\theta_\infty)G_q(1-2\theta_0)}. 
\end{align*}

The following braiding relation holds. 
\begin{thm}\label{thm-braiding}
We have 
\begin{align}\label{braiding-rel}
&
V'\Bigl({{\kern-3pt\half\phantom{0}}\atop {\theta_\infty\phantom{h0}\theta_\infty+\textstyle{\frac{\e'}{2}}}};x_1\Bigr)
V\Bigl({\phantom{00}{\theta_1}\atop {\theta_\infty+\textstyle{\frac{\e'}{2}}
\phantom{h0}\theta_0}};x_2\Bigr)
\\
&=
\sum_{\e=\pm}
V\Bigl({{\theta_1\phantom{00}}\atop {\theta_\infty\phantom{h0}\theta_0+\textstyle{\frac{\e}{2}}}}
;x_2\Bigr)
V'\Bigl({{\phantom{0}\half}\atop {\theta_0+\textstyle{\frac{\e}{2}}\phantom{h0}\theta_0}};x_1\Bigr)
\B_{\e,\e'}\left[\begin{matrix}
		  \theta_1 & \half \\
                  \theta_\infty & \theta_0\\
		 \end{matrix}
\Bigl| \frac{x_2}{x_1}\right]\cdot
q^{\theta_1^2-\theta_1/2}\left(\frac{x_2}{x_1}\right)^{\theta_1}\,,
\nn
\end{align}
where the braiding matrix is given by 
\begin{align}
&\B_{\e,\e'}\left[\begin{matrix}
		  \theta_1 & \half \\
                  \theta_\infty & \theta_0\\
		 \end{matrix}
\Bigl| x\right]
=-\e 
\frac{\vartheta\bigl(\frac{1}{2}+\e'\theta_\infty+\theta_1-\e \theta_0\bigr)}
{\vartheta\bigl(2\theta_0\bigr)}
\frac{\vartheta\bigl(\frac{1}{2}+\e'\theta_\infty+\theta_1+\e \theta_0+u\bigr)}
{\vartheta\bigl(2\theta_1+u\bigr)}
\label{braid}
\end{align}
with $x=q^u$.
\qed
\end{thm}

For the vacuum matrix element, 
the relation \eqref{braiding-rel} is a consequence of
the known connection formula for basic hypergeometric functions.
At this writing it is not clear to us whether the general case
can be deduced from this and the intertwining relation.
In Appendix we give a direct combinatorial proof of 
the braiding relation \eqref{braiding-rel}.

The braiding matrix \eqref{braid} has the following properties. 
\begin{align}
&\B_{\e,\e'}\left[\begin{matrix}
		  \theta_1 & \half \\
                  \theta_\infty & \theta_0\\
		 \end{matrix}
\Bigl| q x\right]
=\B_{\e,\e'}\left[\begin{matrix}
		  \theta_1 & \half \\
                  \theta_\infty & \theta_0\\
		 \end{matrix}
\Bigl| x\right]\,,
\label{braid-q-per}\\
&\B\left[\begin{matrix}
  \theta_1 & \half \\
  \theta_\infty & \theta_0\\
 \end{matrix}
\Bigl| x\right]^{-1}
=
\B\left[\begin{matrix}
  \theta_1 & \half \\
  \theta_0 & \theta_\infty\\
 \end{matrix}
\Bigl| q^{-2\theta_1}x^{-1}\right],
\label{braid-inv}
\\
&
\det \B\left[\begin{matrix}
		  \theta_1 & \half \\
                  \theta_\infty & \theta_0\\
		 \end{matrix}
\Bigl| x\right]
 =\frac{\vartheta\bigl(2\theta_\infty\bigr)}{\vartheta\bigl(2\theta_0\bigr)}
\frac{\vartheta\bigl(u\bigr)}{\vartheta\bigl(u+2\theta_1\bigr)}\,,
\nn
\\
&\B\left[
\begin{matrix}
  \theta_1 & \half \\
  \theta_\infty & \theta_0\\
\end{matrix}
\Bigl| x\right]
\quad
\text{is $1$-periodic in $\theta_0$,$\theta_1$,$\theta_\infty$}\,.
\label{braid-period}
\end{align}

\section{$q$-Painlev{\'e VI} equation}

\subsection{Riemann problem}

This paper deals with a nonlinear nonautonomous $q$-difference equation
called $q$-Painlev\'e VI ($q$PVI) equation.
It is derived from deformation
theory of linear $q$-difference equations as a discrete analog of isomonodromic
deformation \cite{JS}.

A remark is in order concerning the name of the equation.
Discrete Painlev\'e equations are classified by rational surface
theory and sometimes called by the name of the surface \cite{S1}.
From that viewpoint, $q$PVI  is an equation corresponding to the surface of
type $A_3^{(1)}$.
Besides, this surface has 
symmetry by the Weyl group $D_5^{(1)}$.
Accordingly $q$PVI is also 
referred to as the $q$-Painlev{\' e} equation 
of type $A_3^{(1)}$ , or the  $q$-Painlev{\' e} equation 
with $D_5^{(1)}$-symmetry.

Regarding discrete Painlev\'e equations,
pioneering researches by Ramani, Grammaticos, and their coworkers are
well-known.
The  $q$PVI equation, 
in the special case when an extra symmetry condition is fulfilled, 
has been studied earlier by them under the name of
discrete Painlev\'e III equation \cite{RGH}.

In this section we review the content of \cite{JS} concerning 
the generalized Riemann problem associated with the $q$PVI equation. 

Let $\theta_0,\theta_t,\theta_1,\theta_\infty\in\C$. 
We set 
\begin{align*}
&R_1=\max(1,|q^{-2\theta_1}|,|tq^{-2\theta_1}|,|tq^{-2\theta_1-2\theta_t}|), \\ 
&R_2=\min(|q^{-1}|,|q^{-2\theta_1-1}|,
|tq^{-2\theta_1-1}|,|tq^{-2\theta_1-2\theta_t-1}|),
\end{align*}
and assume that $R_1<R_2$. 
Fix also $\sigma\in\C$ and $s\in\C^\times$.

Suppose we are given $2\times 2$ matrices 
$ Y^\infty(x,t),Y^{0t}(x,t),Y^0(x,t)$ with the following properties.
\begin{enumerate}
\item
These matrices are holomorphic in the domains
\begin{align}
Y^\infty(x,t): R_1<|x|\,,
\quad  
Y^{0t}(x,t): R_1<|x|<R_2\,,
\quad  
Y^0(x,t): |x|<R_2\,.
\label{domain}
\end{align}
\item They are related to each other by 
\begin{align}
&Y^\infty(x,t)=Y^{0t}(x,t)\B_1(x),
\quad 
Y^{0t}(x,t)=Y^0(x,t)\B_2(x), 
\label{connection1}
\\
&\B_1(x)
=\B\left[\begin{matrix}
		   \theta_1  &\half\\
                 \theta_\infty+\half & \sigma \\
		 \end{matrix}
\Bigl|q^{-2\theta_1} x^{-1} \right]\,,
\quad 
\B_2(x)
=
\B\left[\begin{matrix}
		 \theta_t& \half  \\
                   \sigma+\half & \theta_0\\
		 \end{matrix}
\Bigl|q^{-2\theta_t-2\theta_1}\frac{t}{x}\right]
\begin{pmatrix}
 0 & s \\ 1 & 0 \\
\end{pmatrix}
\,,
\label{connection2}
\end{align}
\item At $x\to \infty$ or $x\to 0$ they have the behavior
\begin{align}
&Y^\infty(x,t)=\left(I+Y_1(t) x^{-1}+O(x^{-2})\right) 
\begin{pmatrix}
 x^{-\theta_\infty} & 0 \\ 0 & x^{\theta_\infty}\\
\end{pmatrix}
\quad (x\to \infty),
\label{Yinf}\\
&Y^0(x,t)=G(t)(I+O(x))
\begin{pmatrix}
 x^{\theta_0} & 0 \\ 0 & x^{-\theta_0}\\
\end{pmatrix}\times x^{-\theta_t-\theta_1}
\quad (x\to 0).
\label{Y0}
\end{align}
\end{enumerate}

It follows that $Y(x,t)=Y^\infty(x,t)$ is a solution of 
linear $q$-difference equations of the form 
\begin{align}
&Y(qx,t)=A(x,t)Y(x,t)\,,\quad 
A(x,t)=\frac{A_2 x^2+A_1(t)x+A_0(t)}{(x-q^{-1})(x-tq^{-2\theta_1-1})}\,,
\label{A}
\\
&Y(x,qt)=B(x,t)Y(x,t)\,,
\quad B(x,t)=\frac{xI+B_0(t)}{x-tq^{-2\theta_t-2\theta_1}}\,.
\label{B} 
\end{align}
We have 
\begin{align*}
&A_2=\mathrm{diag}(q^{-\theta_\infty},q^{\theta_\infty}),
\quad
A_0(t)=tq^{-3\theta_1-\theta_t-2}
G(t)\cdot \mathrm{diag}(q^{\theta_0},q^{-\theta_0})
\cdot G(t)^{-1},
\\
&\det A(x,t)=\frac{(x-q^{-2\theta_1-1})(x-t q^{-2\theta_t-2\theta_1-1})}
{(x-q^{-1})(x-tq^{-2\theta_1-1})}\,.
\end{align*}
Though the argument is quite standard,   
we outline the derivation for completeness.  

We take the determinant of the relations \eqref{connection1}, \eqref{connection2} 
and rewrite them as 
\begin{align*}
&
\frac{(1/x,tq^{-2\theta_1}/x;q)_\infty}
{(q^{-2\theta_1}/x,tq^{-2\theta_t-2\theta_1}/x;q)_\infty}
\det Y^{\infty}(x,t)\,
\\
&=
-\frac{\vartheta(2\theta_\infty)}{\vartheta(2\sigma)}
\frac{(q^{2\theta_1+1}x,q^{-2\theta_1}t/x;q)_\infty}
{(qx,tq^{-2\theta_t-2\theta_1}/x;q)_\infty}
q^{2\theta_1^2+\theta_1}x^{2\theta_1}
\det Y^{0t}(x,t)
\\
&=-s\frac{\vartheta(2\theta_\infty)}{\vartheta(2\theta_0)}
\frac{(q^{2\theta_1+1}x,q^{2\theta_t+2\theta_1+1}x/t;q)_\infty}
{(qx,q^{2\theta_1+1}x/t;q)_\infty}
q^{2\theta_t^2+\theta_t+2\theta_1^2+\theta_1}
\left(q^{2\theta_1}\frac{x}{t}\right)^{2\theta_t}x^{2\theta_1}
\det Y^0(x,t)\,.
\end{align*}
Both sides of this expression are single-valued
and holomorphic on $\C^\times$. 
In view of the  behavior \eqref{Yinf}, \eqref{Y0} we conclude that 
\begin{align*}
\det Y^{\infty}(x,t)
 =
\frac{(q^{-2\theta_1}/x,tq^{-2\theta_t-2\theta_1}/x;q)_\infty}
{(1/x,tq^{-2\theta_1}/x;q)_\infty}\,.
\end{align*}
In particular, the matrices $Y^\infty(x,t)$, $Y^{0t}(x,t)$ and $Y^0(x,t)$ 
are generically invertible. 
Rewriting \eqref{connection1}, \eqref{connection2}
and using the periodicity of the braiding matrices, we obtain
\begin{align*}
A(x,t):=Y^\infty(qx,t){Y^\infty(x,t)}^{-1}
=Y^{0t}(qx,t){Y^{0t}(x,t)}^{-1} 
=Y^0(qx,t){Y^0(x,t)}^{-1}\,. 
\end{align*}
A similar argument shows that $A(x,t)$ 
is a rational function of $x$
with the only poles at $x=q^{-1},q^{-2\theta_1-1}t$. 
Using \eqref{Yinf}, \eqref{Y0} 
once again we find that $A(x,t)$ has the form 
stated above. 
The derivation of \eqref{B} is quite similar. 

Introduce  $y=y(t)$, $z=z(t)$ and $w=w(t)$ by
\begin{align*}
&A(x,t)_{+-} =\frac{
q^{\theta_\infty}w(x-y)
}
{(x-q^{-1})(x-tq^{-2\theta_1-1})}\,,
\\
&A(y,t)_{++}=\frac{y-tq^{-2\theta_t-2\theta_1-1}}{qz(y-q^{-1})}
\,,
\end{align*}
where we use $\pm$ to indicate components of 
a $2\times2$ matrix.
Notation being as above, the 
compatibility
\begin{equation}
 A(x,qt)B(x,t)=B(qx,t)A(x,t)\label{compatibility}
\end{equation}
of \eqref{A} and \eqref{B} leads to the $q$PVI equation. 
Here and after we write $\bar{f}(t)=f(qt)$,  $\underline{f}(t)=f(q^{-1}t)$.  

\begin{prop}\cite{JS}
The functions $y,z$ solve the $q$PVI system
\begin{align*}
&\frac{y\overline{y}}{a_3a_4}=
\frac{(\overline{z}-tb_1)(\overline{z}-tb_2)}{(\overline{z}-b_3)(\overline{z}-b_4)}\,,
\quad 
\frac{z\overline{z}}{b_3b_4}=
\frac{(y-ta_1)(y-ta_2)}{(y-a_3)(y-a_4)}\,,
\end{align*}
with the parameters
\begin{align*}
&a_1=q^{-2\theta_1-1},\quad a_2=q^{-2\theta_t-2\theta_1-1} ,
\quad a_3=q^{-1},\quad a_4=q^{-2\theta_1-1},\\
&b_1=q^{-\theta_0-\theta_t-\theta_1},\quad 
b_2=q^{\theta_0-\theta_t-\theta_1},\quad 
\quad b_3=q^{\theta_\infty-1},\quad 
b_4=q^{-\theta_\infty}\,.
\end{align*}
\qed
\end{prop}

\begin{rem}
 From the compatibility condition \eqref{compatibility},
 we obtain the expression of the matrix $B(x,t)$ in terms of $y$ and
 $z$, etc.
 In particular, the $(+,-)$  
element of the matrix $B_0(q^{-1}t)$ is
 written as
 \begin{equation}\label{B_z}
  B_0(q^{-1}t)_{+-}=\frac{q^{1+\theta_\infty}zw}{1-q^{1-\theta_\infty}z}.
 \end{equation}
\end{rem}

\subsection{CFT construction}
We are now in a position to construct solutions of the 
generalized Riemann problem in terms of conformal block functions. 
Following the method of \cite{ILTe}, we consider the following 
sums of 5 point conformal blocks, 
\begin{align}
&Y^{\infty}_{\e,\e'}(x,t)=
\frac{1}{k_\infty(\e)} 
\sum_{n\in\Z}s^n
\F\Bigl({{\phantom{hh00}\half
\phantom{hhhh0000} \theta_1\phantom{hhhhh}
\theta_t 
\phantom{hhhhh}}
\atop
{\theta_\infty-\textstyle{\frac{\e}{2}}
\phantom{h0}\theta_\infty-\textstyle{\frac{\e-\e'}{2}}
\phantom{h0}\sigma+n
\phantom{h0}
\theta_0
}}
;q^{2\theta_t+2\theta_1}x,q^{2\theta_t},t\Bigr) \,,
\label{CFT-Yinf}\\
&Y^{0t}_{\e,\e'}(x,t)
=\frac{x^{-\theta_1}}{k_{0t}(\e)} 
\sum_{n\in\Z}s^n
\F\Bigl({{\phantom{0000}\theta_1
\phantom{hhhhh00000}\half
\phantom{hhhhh} \theta_t
\phantom{hhhhh00}}
\atop
{\theta_\infty-\textstyle{\frac{\e}{2}}
\phantom{h0}\sigma+n+\textstyle{\frac{\e'}{2}}
\phantom{h0}\sigma+n
\phantom{h0}\theta_0
}
}
;q^{2\theta_t},q^{2\theta_t+2\theta_1}x,t\Bigr) 
\label{CFT-Y0t}\\
&\qquad \qquad
=\frac{x^{-\theta_1}}{k_{0t}(\e)} 
\sum_{n\in\Z}s^{n+\frac{1-\e'}{2}}
\F\Bigl({{\phantom{00000} \theta_1\phantom{hhhhhhh00}\half\phantom{hhhhhhhh000}\theta_t\phantom{hhh}}
\atop
{\theta_\infty-\textstyle{\frac{\e}{2}}
\phantom{h0}\sigma+n+\half
\phantom{h0}\sigma+n+\textstyle{\frac{1-\e'}{2}}
\phantom{h0}\theta_0
}}
;q^{2\theta_t},q^{2\theta_t+2\theta_1}x,t\Bigr)\,,
\nn\\
&Y^0_{\e,\e'}(x,t)=
\frac{x^{-\theta_t-\theta_1}}{k_0(\e)} 
\sum_{n\in\Z}s^n
\F\Bigl(
{{\phantom{00000}\theta_1\phantom{hhhhh00} \theta_t\phantom{hhhh00}\half\phantom{hhhh}}
\atop
{\theta_\infty-\textstyle{\frac{\e}{2}}
\phantom{h0}
\sigma+n+\half
\phantom{h0}
\theta_0+\textstyle{\frac{\e'}{2}}
\phantom{h0}
\theta_0}}
;q^{2\theta_t},t,q^{2\theta_t+2\theta_1}x\Bigr) \,,
\label{CFT-Y0}
\end{align}
where  
\begin{align*}
&k_\infty(\varepsilon)=q^{(\theta_\infty-\varepsilon/2)^2-2\varepsilon\theta_\infty(\theta_t+\theta_1)}
\N'\Bigl({{\phantom{h} \half \phantom{h_3}}
\atop{\theta_\infty-\textstyle{\frac{\e}{2}} \phantom{h_2} 
\theta_\infty   
}}
\Bigr)
\hat{\tau}\,,
\\
&k_{0t}(\e)=k_\infty(\e) q^{\theta_1^2+\theta_1/2}, 
\quad
k_{0}(\e)=k_{0t}(\e) 
q^{\theta_t^2+\theta_t/2+2\theta_1\theta_t}t^{-\theta_t}, 
\end{align*}
\\
and 
\begin{align*}
&\hat{\tau}=
\sum_{n\in\Z}s^n
\F\Bigl({{\phantom{h00}\theta_1 \phantom{hhhh} \theta_t\phantom{hhhh}}
\atop
{\theta_\infty\phantom{h0}\sigma+n\phantom{h0}\theta_0}}
;q^{2\theta_t},t\Bigr) \,.
\end{align*}
We assume that these series converge in the domains \eqref{domain}. 
The asymptotic behavior \eqref{Yinf}, \eqref{Y0} 
are simple consequences of the expansion \eqref{block} of conformal blocks.
The crucial point is the validity of the connection formulas 
\eqref{connection1}, \eqref{connection2}
due to the periodicity property \eqref{braid-period} of the braiding matrix. 
Therefore, formulas \eqref{CFT-Yinf}, \eqref{CFT-Y0t}, \eqref{CFT-Y0}
solve the generalized Riemann problem.

Define the tau function by 
\begin{align*}
&\tau\left[
\begin{matrix}
 \theta_1  & \theta_t \\
 \theta_\infty  & \theta_0\\
\end{matrix}\Bigl|s,\sigma,t\right]
=q^{-2(\theta_t+\theta_1)\theta_\infty^2+2\theta_t\theta_1^2}
G_q(1+2\theta_\infty)G_q(1-2\theta_0)\cdot
\hat{\tau} \,.
\end{align*}
Explicitly it 
has the AGT series representation
\begin{align}
&\tau\left[
\begin{matrix}
 \theta_1  & \theta_t \\
 \theta_\infty  & \theta_0\\
\end{matrix}\Bigl|s,\sigma,t\right]
=\sum_{n\in\Z}
s^n t^{(\sigma+n)^2-\theta_t^2-\theta_0^2}
C\left[\begin{matrix}
\theta_1   & \theta_t \\
 \theta_\infty &  \theta_0\\
\end{matrix}\Bigl|\sigma+n\right]
Z\left[\begin{matrix}
\theta_1  &   \theta_t\\
\theta_\infty   &\theta_0 \\
\end{matrix}\Bigl|\sigma+n,t\right]\,,
\label{taufn}
\end{align}
with the definition
\begin{align*}
&C\left[\begin{matrix}
\theta_1  &  \theta_t \\
\theta_\infty &  \theta_0 \\
\end{matrix}\Bigl|
\sigma\right]
=\frac{\prod_{\e,\e'=\pm}G_q(1+\e\theta_\infty-\theta_1+\e'\sigma)
G_q(1+\e\sigma-\theta_t+\e'\theta_0)}
{G_q(1+2\sigma)G_q(1-2\sigma)}
\,,
\\
&Z\left[\begin{matrix}
 \theta_1  &  \theta_t\\
 \theta_\infty &  \theta_0\\
\end{matrix}\Bigl|
\sigma,t\right]
=\sum_{\bla=(\lambda_+,\lambda_-)\in\Lambda^2}
t^{|\bla|}\cdot
\frac{
\prod_{\e,\e'=\pm}
N_{\emptyset,\lambda_{\e'}}(q^{\e\theta_\infty-\theta_1-\e'\sigma})
N_{\lambda_{\e},\emptyset}(q^{\e\sigma-\theta_t-\e'\theta_0})}
{\prod_{\e,\e'=\pm}N_{\lambda_\e,\lambda_{\e'}}(q^{(\e-\e')\sigma})}
\,.
\end{align*}

Comparing the asymptotic expansion at $x=0$ and $x=\infty$, 
we can express
the $(+,-)$
element of the matrix $A(x,t)$ 
and $B(x,t)$
in terms of tau functions. 
\begin{thm}
The following expressions hold for $y$,
$z$, and $w$
in terms of tau functions:
\begin{align}
y=&q^{-2\theta_1-1}t\cdot \frac{\tau_3\tau_4}{\tau_1\tau_2}\,,
\label{y-tau}\\
z=&\frac{\underline{\tau_1}\tau_2-\tau_1\underline{\tau_2}}
{q^{\theta_\infty}\underline{\tau_1}\tau_2-q^{1-\theta_\infty}\tau_1\underline{\tau_2}}\,,
\label{z-tau}\\
w=&q^{-1}(1-q^{1-2\theta_\infty})
\frac{\Gamma_q(2\theta_\infty)}
{\Gamma_q(2-2\theta_\infty)}\cdot
\frac{\tau_2}{\tau_1}\,.\nn
\end{align}
Here we put
\begin{alignat*}{2}
\tau_1 &=\tau\left[
\begin{matrix}
  \theta_1 & \theta_t \\
  \theta_\infty & \theta_0\\
\end{matrix}\Bigl|s,\sigma,t\right]\,,
&
\tau_2 &=\tau\left[
\begin{matrix}
  \theta_1 & \theta_t \\
  \theta_\infty-1 & \theta_0\\
\end{matrix}\Bigl|s,\sigma,t\right]\,,
\\
\tau_3 &=\tau\left[
\begin{matrix}
  \theta_1 & \theta_t \\
  \theta_\infty-\half & \theta_0+\half\\
\end{matrix}\Bigl|s,\sigma+\half ,t\right]\,,
\quad
&
\tau_4 &=\tau\left[
\begin{matrix}
  \theta_1 & \theta_t \\
  \theta_\infty-\half & \theta_0-\half\\
\end{matrix}\Bigl|s,\sigma-\half ,t\right]\,.
\end{alignat*}
\qed
\end{thm}

\begin{proof}
From the explicit expression of the 5-point conformal block function
 \eqref{block}, we have
 \begin{align*}
  \F\Bigl(
 {
 {\ \theta_3\
 \theta_{2}
 \ \theta_1\ }
 \atop
 {\theta_{4}\ \sigma_{2}\
 \sigma_{1}\ \theta_{0}}
 }
 ;x_3,x_2,x_{1}
 \Bigr)
 =&\
 \N\Bigl({\kern-1pt\theta_1\atop{\sigma_1\ \theta_0}}
 \Bigr)q^{2\theta_1\sigma_{1}^2}
 \cdot 
 x_1^{\sigma_{1}^2-\theta_1^2-\theta_{0}^2}
 \Bigl(
 \F\Bigl(
 {
 {\ \theta_3
 \ \theta_2\ }
 \atop
 {\theta_{4}\
 \sigma_{2}\ \sigma_{1}}
 }
 ;x_3,x_2
 \Bigr) +O(x_1)
 \Bigr)\\
 =&\
 \N\Bigl({\kern-1pt\theta_3\atop{\theta_4\ \sigma_2}}
 \Bigr)q^{2\theta_3\theta_{4}^2}
 \cdot 
 x_3^{\theta_{4}^2-\theta_3^2-\sigma_{2}^2}
 \Bigl(
 \F\Bigl(
 {
 {\ \theta_2
 \ \theta_1\ }
 \atop
 {\sigma_{2}\ \sigma_{1}\
 \theta_0}
 }
 ;x_2,x_1
 \Bigr) +O\!\left(\frac{1}{x_3}\right)
 \Bigr)\,.
 \end{align*}
 By using this expression,
 we calculate the asymptotic behavior of $Y^\infty(x,t)$, $Y^0(x,t)$ and get
 \begin{align*}
  G(t)_{\e ,\e'}=q^Ct^{\theta_t}
  \frac{\Gamma_q(2\e\theta_\infty)}{\Gamma_q(1+2\e'\theta_0)}
  \frac{\tau\left[
  \begin{matrix}
  \theta_1 & \theta_t \\
  \theta_\infty-\frac{\e}{2} & \theta_0+\frac{\e'}{2}\\
  \end{matrix}
  \Bigl|s,\sigma+\half,t\right]}{\tau_1},\qquad
  Y_1(t)_{+-}=\frac{\Gamma_q(2\theta_\infty)}{\Gamma_q(2-2\theta_\infty)}
  \cdot\frac{\tau_2}{\tau_1},
 \end{align*}
 where $C={\theta_0}^2-(\theta_1+\theta_t)^2-{\theta_\infty}^2
 +\e' (2\theta_1+2\theta_t+1) +\e\theta_\infty$ (see
 \eqref{Yinf}--\eqref{Y0} for $G(t)$ and $Y_1(t)$).

 We can write down the exponents of matrices $A(x,t)$ and $B(x,t)$ in
 these terms.
 In particular, the expressions
 \begin{equation*}
  A_0(t)_{+-}=q^{-3\theta_1-\theta_t-2}(q^{-\theta_0}-q^{\theta_0})
  \frac{G(t)_{++}G(t)_{+-}}{\det G(t)},\qquad
  A_1(t)_{+-}=(q^{\theta_\infty-1}-q^{-\theta_\infty})Y_1(t)_{+-}
 \end{equation*}
 give the formula for $w$ and $y$.
 The expression of $z$ is obtained from
 \begin{equation*}
   B_0(q^{-1}t)_{+-}=Y_1(t)_{+-}-Y_1(q^{-1}t)_{+-}
 \end{equation*}
 and the relation \eqref{B_z}.
\end{proof}

In \cite{Ma}, Mano constructed two parametric local solutions of 
the $q$PVI equation near the critical points $t=0,\infty$. 
The above formula for $y$ extends Mano's asymptotic expansion to all orders. 

We expect that the formula \eqref{z-tau}
for $z$ can be simplified, 
see \eqref{yz-final} below.


\subsection{Bilinear equations}

In \cite{TM}, Tsuda and Masuda formulated the tau function of the 
$q$PVI equation to be a function on the weight lattice of 
$D^{(1)}_5$ with symmetries under the affine Weyl group.
Motivated by \cite{TM},
we consider the set of 
eight tau functions defined by the AGT series \eqref{taufn}:
\begin{alignat*}{2}
\tau_1 &=\tau\left[
\begin{matrix}
  \theta_1 & \theta_t \\
  \theta_\infty+\half & \theta_0\\
\end{matrix}\Bigl|s,\sigma,t\right]\,,
&
\tau_2 &=\tau\left[
\begin{matrix}
  \theta_1 & \theta_t \\
  \theta_\infty-\half & \theta_0\\
\end{matrix}\Bigl|s,\sigma,t\right]\,,
\\
\tau_3 &=\tau\left[
\begin{matrix}
  \theta_1 & \theta_t \\
  \theta_\infty & \theta_0+\half\\
\end{matrix}\Bigl|s,\sigma+\half ,t\right]\,,
&
\tau_4 &=\tau\left[
\begin{matrix}
  \theta_1 & \theta_t \\
  \theta_\infty & \theta_0-\half\\
\end{matrix}\Bigl|s,\sigma-\half ,t\right]\,,
\\
\tau_5 &=\tau\left[
\begin{matrix}
  \theta_1-\half & \theta_t \\
  \theta_\infty & \theta_0\\
\end{matrix}\Bigl|s,\sigma,t\right]\,,
&
\tau_6 &=\tau\left[
\begin{matrix}
  \theta_1+\half & \theta_t \\
  \theta_\infty & \theta_0\\
\end{matrix}\Bigl|s,\sigma,t\right]\,,
\\
\tau_7 &=\tau\left[
\begin{matrix}
  \theta_1 & \theta_t-\half \\
  \theta_\infty & \theta_0\\
\end{matrix}\Bigl|s,\sigma+\half ,t\right]\,,
\quad
&
\tau_8 &=\tau\left[
\begin{matrix}
  \theta_1 & \theta_t+\half \\
  \theta_\infty & \theta_0\\
\end{matrix}\Bigl|s,\sigma-\half ,t\right]\,.
\\
\end{alignat*}
For convenience, we shift the parameter $\theta_\infty$ 
in the previous section to 
$\theta_\infty+1/2$. 

On the basis of a computer run we put forward: 
\begin{conj}
For the eight tau functions defined above, the following 
bilinear equations hold. 
\begin{align}
&\tau_1\tau_2-q^{-2\theta_1}t \,\tau_3\tau_4
-(1-q^{-2\theta_1}t)\tau_5 \tau_6=0,
\label{bilin-1}\\ 
&\tau_1\tau_2-t \,\tau_3\tau_4
-(1-q^{-2\theta_t}t)\,\underline{\tau_5}\overline{\tau_6}=0,
\label{bilin-2}\\ 
 &\tau_1\tau_2-\tau_3\tau_4+(1-q^{-2\theta_1}t)q^{2\theta_t}\underline{\tau_7}\overline{\tau_8}=0,
\label{bilin-3}\\ 
&\tau_1\tau_2-q^{2\theta_t}\tau_3\tau_4+
(1-q^{-2\theta_t}t)q^{2\theta_t}\tau_7 \tau_8=0,
\label{bilin-4}\\ 
&\underline{\tau_5}\tau_6+q^{-\theta_1-\theta_\infty+\theta_t-1/2}t\, \underline{\tau_7}\tau_8
-\underline{\tau_1}\tau_2=0,
\label{bilin-5}\\ 
&\underline{\tau_5}\tau_6+q^{-\theta_1+\theta_\infty+\theta_t-1/2}t\, \underline{\tau_7}\tau_8
-\tau_1\underline{\tau_2}=0,
\label{bilin-6}\\
&\underline{\tau_5}\tau_6+q^{\theta_0+2\theta_t}\underline{\tau_7}\tau_8
-q^{\theta_t} \underline{\tau_3}\tau_4=0,
\label{bilin-7}\\ 
&\underline{\tau_5}\tau_6+q^{-\theta_0+2\theta_t}\underline{\tau_7}\tau_8
-q^{\theta_t} \tau_3\underline{\tau_4}=0\, .
\label{bilin-8}
\end{align} 
The solution $y,z$ of qPVI are expressed as 
\begin{align}\label{yz-final}
y=q^{-2\theta_1-1}t\cdot \frac{\tau_3\tau_4}{\tau_1\tau_2}\,,
\quad
z=-q^{\theta_t-\theta_1-1}  
t\cdot
\frac{\underline{\tau_7}\tau_8}{\underline{\tau_5}\tau_6}\,.
\end{align}
\qed
\end{conj}

The bilinear equations \eqref{bilin-1}--\eqref{bilin-8} 
were originally proved in \cite{S} in the special case 
when the tau functions are given by 
Casorati determinants of basic hypergeometric functions. 
It would be interesting if one can prove 
them in the general case 
by a proper extension of the methods of \cite{BS} or \cite{GL}.

\begin{rem}
Formula \eqref{z-tau} for $z$ is obtained from the asymptotic expansion at $x=\infty$.
The expansion at $x=0$ gives an alternative formula for the same quantity.
Comparing these (and shifting $\theta_\infty$ to 
$\theta_\infty+1/2$),
we arrive at the following bilinear relation:
\begin{align*}
 \underline{\tau_1}\tau_2-\tau_1\underline{\tau_2}
 =\frac{q^{1/2+\theta_\infty}-q^{1/2-\theta_\infty}}
 {q^{-\theta_0}-q^{\theta_0}}q^{-\theta_1-1}t
 \left(\underline{\tau_3}\tau_4-\tau_3\underline{\tau_4}\right)\,.
\end{align*}
This is consistent with 
the conjectured relations \eqref{bilin-5}--\eqref{bilin-8}.
\end{rem}

\appendix

\section{Braiding relation}

In this appendix, we give a direct proof of Theorem \ref{thm-braiding}. 
\medskip

We rewrite the braiding relation 
\eqref{braiding-rel}
into the form of matrix elements 
$\langle \alpha|\otimes \langle\beta|\cdots |\lambda\rangle\otimes|\mu\rangle$ 
for any partitions $\lambda,\mu,\alpha,\beta$. 
The matrix element
\begin{equation*}
\langle \alpha|\otimes \langle\beta
|V'\Bigl({{\kern-3pt\half\phantom{0}}\atop {\theta_\infty\phantom{h0}\theta_\infty+\textstyle{\frac{\e'}{2}}}};x_1\Bigr)
V\Bigl({\phantom{00}{\theta_1}\atop {\theta_\infty+\textstyle{\frac{\e'}{2}}
\phantom{h0}\sigma
}};x_2\Bigr)
|\lambda\rangle\otimes|\mu\rangle
\end{equation*}
can be read off from 
the coefficient of 
$t^{|\lambda|+|\mu|}x_3^{-|\alpha|-|\beta|}$ 
in the 6 point conformal block
\begin{align*}
\F\Bigl(
{
{\phantom{00} \theta_2\phantom{h000}
\frac{1}{2}\phantom{h0}\phantom{h0}\theta_1
\phantom{h0}
\phantom{h0}\theta_t\phantom{h0}}
\atop
{\theta_3\phantom{h0}\theta_\infty
\phantom{h0}\theta_\infty+\frac{\e'}{2}\phantom{h0}
\phantom{h_0}
\sigma\phantom{h0}\theta_{0}}
}
;x_3,x_1,x_2,
t
\Bigr)\,.  
\end{align*}  
Formula \eqref{block} tells that 
it is a sum over a pair of partitions. 
By using $N_{\lambda,\mu}(1)=\delta_{\lambda,\mu}N_{\lambda,\lambda}(1)$,  
the sum can be reduced further to one over a single partition.
We thus find that 
 Theorem \ref{thm-braiding} is equivalent to the following proposition.  

\begin{prop}\label{prop-braiding-relation}
For any partitions $\lambda$, $\mu$, $\alpha$ 
and $\beta$, we have 
\begin{align}
&q^{\theta_\infty^2}\left(\frac{q^{2\theta_1} x_2}{x_1}\right)^{\e'\theta_\infty+1/4}
\frac{
\prod_{\e''=\pm}
\Gamma_q(\half+\e'\theta_\infty-\theta_1+\e'' \sigma)
}{\Gamma_q(1+2\e'\theta_\infty)}X_{\lambda,\mu,\alpha,\beta}
^{\e'}(\theta_\infty,\theta_1,\sigma;x_1,x_2)\label{eq-general-partitions}
\\
&=\sum_{\e=\pm}q^{\sigma^2}\left(\frac{q x_1}{x_2}\right)^{\e\sigma+1/4}
\frac{
\prod_{\e''=\pm}\Gamma_q(\half+\e''
\theta_\infty-\theta_1+\e \sigma)
}{\Gamma_q(1+2\e\sigma)}
 Y_{\lambda,\mu,\alpha,\beta}
^{\e}(\theta_\infty,\theta_1,\sigma;x_1,x_2)\nn
\\
& \times \B_{\e,\e'}\left[\begin{matrix}
		  \theta_1 & \half \\
                  \theta_\infty & \sigma\\
		 \end{matrix}
\Bigl| \frac{x_2}{x_1}\right]
q^{\theta_1^2-\theta_1/2}\left(\frac{x_2}{x_1}\right)^{\theta_1}, \nn
\end{align}
where 
\begin{align*}
&X_{\lambda,\mu,\alpha,\beta}
^{+}(\theta_\infty,\theta_1,\sigma;x_1,x_2)
\\
&=N_{\alpha,\beta}(q^{2\theta_\infty})
N_{\beta,\lambda}
(q^{-\theta_\infty-1/2-\theta_1-\sigma})
N_{\beta,\mu}
(q^{-\theta_\infty-1/2-\theta_1+\sigma})
\\
&\times \sum_{\eta\in\Lambda}
\left(\frac{1}{x_2}\right)^{
|\lambda|+|\mu|}
\left(\frac{q^{2\theta_1} x_2}{x_1}\right)^{|\eta|+|\beta|}
(qx_1)^{|\alpha|+|\beta|
}
\frac{
N_{\alpha,\eta}(q^{-1})
N_{\eta,\lambda}
(q^{\theta_\infty+1/2-\theta_1-\sigma})
N_{\eta,\mu}
(q^{\theta_\infty+1/2-\theta_1+\sigma})}
{N_{\eta,\beta}
(q^{2\theta_\infty+1})
N_{\eta,\eta}(1)},\nonumber
\\
&X_{\lambda,\mu,\alpha,\beta}
^{-}(\theta_\infty,\theta_1,\sigma;x_1,x_2)=X_{\lambda,\mu,\beta,\alpha}
^{+}(-\theta_\infty,\theta_1,\sigma;x_1,x_2), 
\\
&Y_{\lambda,\mu,\alpha,\beta}
^{+}(\theta_\infty,\theta_1,\sigma;x_1,x_2)
\\
&=N_{\lambda,\mu}(q^{2\sigma})
N_{\alpha,\lambda}
(q^{\theta_\infty-\theta_1-\sigma-1/2})
N_{\beta,\lambda}
(q^{-\theta_\infty-\theta_1-\sigma-1/2})
\\
&\times 
\sum_{\eta\in\Lambda}
\left(\frac{1}{x_1}\right)^
{|\lambda|+|\mu|}
\left(\frac{q x_1}{x_2}\right)^{|\lambda|+|\eta|}
(q^{2\theta_1}x_2)^{|\alpha|+|\beta|}
\frac{
N_{\eta,\mu}(q^{-1})
N_{\alpha,\eta}
(q^{\theta_\infty-\theta_1+\sigma+1/2})
N_{\beta,\eta}
(q^{-\theta_\infty-\theta_1+\sigma+1/2})
}
{N_{\lambda,\eta}(q^{2\sigma+1})
N_{\eta,\eta}(1)},
\\
&Y_{\lambda,\mu,\alpha,\beta}
^{-}(\theta_\infty,\theta_1,\sigma;x_1,x_2)=
Y_{\mu,\lambda,\alpha,\beta}
^{+}(\theta_\infty,\theta_1,-\sigma;x_1,x_2).
\\
\end{align*}
\qed
\end{prop}
In the following, we show that the 
identity \eqref{eq-general-partitions} is 
reduced to the connection formula of basic hypergeometric functions. 
First, we prepare several lemmas.

\begin{lem}\label{lem-Nfactor-symmetry}
For any $\lambda, \mu\in\Lambda$ 
we have $N_{\lambda',\mu'}(u)=N_{\mu,\lambda}(u)$.
\qed
\end{lem}
\begin{proof}
Since 
$a_{\lambda'}(j,i)=\ell_\lambda(i,j)$ for all $(i,j)\in\Z^2_{>0}$, we have
\begin{align*}
N_{\lambda', \mu'}(u)
=\prod_{(i,j)\in\lambda} (1-q^{-a_\lambda(i,j)-\ell_\mu(i,j)-1}u)
\prod_{(i,j)\in\mu} (1-q^{a_\mu(i,j)+\ell_\lambda(i,j)+1}u). 
\end{align*}
Therefore, to prove the Lemma it suffices to show that the identity
\begin{equation}\label{eq-Nfactor-sym}
\prod_{(i,j)\in\lambda, (i,j)\notin \mu} (1-q^{-a_\lambda(i,j)-\ell_\mu(i,j)-1}u)=
\prod_{(i,j)\in\lambda, (i,j)\notin \mu} (1-q^{a_\mu(i,j)+\ell_\lambda(i,j)+1}u)  
\end{equation}
holds for all $\lambda,\mu\in\Lambda$. 
Note that  the product consists of all entries ($i,j$) such that 
$\lambda_i>\mu_i$ and $\mu_i+1\le j\le \lambda_i$.  
We show \eqref{eq-Nfactor-sym}
by induction with respect to the length of $\lambda$. 

Consider a sub-sequence $\lambda_{n+1},\lambda_{n+2},\ldots, \lambda_{n+s}$ in $\lambda$ such that 
\begin{equation*}
\lambda_{n}\le \mu_{n},\quad 
\lambda_i>\mu_i\quad (i=n+1,\ldots,n+s),\quad 
\lambda_{n+s+1}\le \mu_{n+s+1}. 
\end{equation*}
Put $\tilde{\lambda}=(\lambda_{n+1},\ldots, \lambda_{n+s})$ and 
$\tilde{\mu}=(\mu_{n+1},\ldots, \mu_{n+s})$. 
Then we have for $i=1,\ldots,s$ and $\tilde{\mu}_i+1\le j\le \tilde{\lambda}_i$
\begin{equation*}
\ell_\mu(n+i,j)=\ell_{\tilde{\mu}}(i,j), \quad \ell_\lambda(n+i,j)=\ell_{\tilde{\lambda}}(i,j), 
\end{equation*}
because $\lambda_{n}\le \mu_{n}$ and $\lambda_{n+s+1}\le \mu_{n+s+1}$. 
We further divide a partition $\tilde{\lambda}$ into sub-sequences 
$\tilde{\lambda}_{m+1},\ldots,\tilde{\lambda}_{m+r}$ such that 
\begin{align*}
\tilde{\lambda}_{m+1}\le \tilde{\mu}_{m},\quad 
\tilde{\lambda}_i>\tilde{\mu}_{i-1}\quad (i=m+2,\ldots,m+r),\quad 
\tilde{\lambda}_{m+r+1}\le \tilde{\mu}_{n+r}. 
\end{align*}
Put $\hat{\lambda}=(\tilde{\lambda}_{m+1},\ldots, \tilde{\lambda}_{m+r})$ and 
$\hat{\mu}=(\tilde{\mu}_{m+1},\ldots, \tilde{\mu}_{m+r})$. Then we have again  
for $i=1,\ldots,r$ and $\hat{\mu}_i+1\le j\le \hat{\lambda}_i$
\begin{equation*}
\ell_{\tilde{\mu}}(m+i,j)=\ell_{\hat{\mu}}(i,j), \quad \ell_{\tilde{\lambda}}(m+i,j)
=\ell_{\hat{\lambda}}(i,j), 
\end{equation*}
because $\tilde{\lambda}_{m+1}\le \tilde{\mu}_{m}$ and 
$\tilde{\lambda}_{m+r+1}\le \tilde{\mu}_{m+r}$.
Replacing $\lambda,\mu$ by $\hat{\lambda},\hat\mu$ 
we may assume without loss of generality that 
\begin{equation}\label{la-mu}
\lambda_1>\mu_1\,,\quad \lambda_i>\mu_{i-1}\quad i=2,\ldots, \ell(\lambda).
\end{equation}

If $\ell(\lambda)=1$, it is easy to check \eqref{eq-Nfactor-sym}.

Suppose 
\eqref{eq-Nfactor-sym} is true for $\ell(\lambda)=\ell-1$. 
Let $\lambda$ satisfy $\ell(\lambda)=\ell$.
Assuming \eqref{la-mu}, 
we choose the number $k$ such that $\mu_{\ell-(k+1)}>\lambda_\ell\ge\mu_{\ell-k}$.

Writing $\lambda=(\nu,\lambda_\ell)$, 
we compare \eqref{eq-Nfactor-sym} for $\lambda$ and for $\nu$.
In the left hand side of \eqref{eq-Nfactor-sym}, 
the factors for $i<\ell$ remain the same as for $\nu$.  
The product of factors for $i=\ell$ is expressed as 
\begin{equation}\label{eq-Nfactor-LHS-l}
\frac{\prod_{j=0}^{\lambda_\ell-\mu_\ell+k-1}\left(1-q^{k-j}u\right)}
{\prod_{j=1}^k(1-q^{-\lambda_\ell+\mu_{\ell-j}+j}u)}.
\end{equation}

In the right hand side of \eqref{eq-Nfactor-sym}, the factors for $j> \lambda_\ell$ 
remain the same as for $\nu$.  
The product of factors for $i=\ell$ is equal to 
\begin{equation*}
\prod_{j=0}^{\lambda_\ell-\mu_\ell-1}\left(1-q^{-j}u\right).
\end{equation*}
The products of factors for $i=\ell-k,\ldots,\ell-1$ and $j\le \lambda_\ell$ are equal to 
\begin{equation*}
\prod_{m=1}^{\lambda_\ell-\mu_i}\left( 1-q^{\ell-i-m}u\right)
\end{equation*}
before adding the $\ell$-th row, and after they become
\begin{equation*}
\prod_{m=0}^{\lambda_\ell-\mu_i-1}\left(1-q^{\ell-i-m}u\right).
\end{equation*} 
Namely, we lose the factor $\prod_{i=\ell-k}^{\ell-1}(1-q^{-\lambda_\ell+\mu_i+\ell-i}u)$   
and gain the factor $\prod_{m=1}^k(1-q^{m}u)$. Therefore, the outcome can be written as 
\eqref{eq-Nfactor-LHS-l}. 
\end{proof}

We introduce the following operations on partitions $\lambda\in \Lambda$:
\begin{align}
\bar{\lambda}=&(\lambda_1-1,\ldots, \lambda_{\ell(\lambda)}-1)\,,  
\label{lambda-bar}\\
r_n(\lambda)=&(\lambda_1+1,\ldots, \lambda_n+1,\lambda_{n+2},\ldots)\quad
(n\in\Z_{\ge0}).
\label{rn}
\end{align}
Here we identify a partition $\lambda$ with a sequence 
$(\lambda_1,\ldots,\lambda_{\ell(\lambda)},0,0,\ldots)$.

\begin{lem}
Let $\lambda,\eta$ be partitions. Then $N_{\eta,\lambda}(q^{-1})\neq 0$ 
if and only if $\eta=r_n(\lambda)$ for some $n\ge 0$. 
\qed
\end{lem}
\begin{proof}
It is easy to show the `if' part.

To show the `only if' part, assume that
$\eta$ satisfies $N_{\eta,\lambda}(q^{-1})\neq 0$. 
Then we have
\begin{align}
&\ell_\eta(i,j)+a_\lambda(i,j)+2\neq0\qquad 
\forall (i,j)\in \eta\,,
\label{eq-eta-lambda-1}\\
&\ell_\lambda(i,j)+a_\eta(i,j)\neq 0\qquad 
\forall (i,j)\in\lambda\,. 
\label{eq-eta-lambda-2}
\end{align}

Using them we show
\begin{enumerate}
\item If $\ell(\lambda)<i\le \ell(\eta)$ then $\eta_i=1$.
\item If $i\le \min(\ell(\lambda),\ell(\eta))$ then $\eta_i\le\lambda_i+1$,
\item If $i\le \min(\ell(\lambda),\ell(\eta))$ and $\eta_i\le\lambda_i$,
then $i<\ell(\lambda)$ and $\eta_i\le \lambda_{i+1}$. 
\item If $i\le \min(\ell(\lambda)-1,\ell(\eta))$, 
then $\eta_i\ge \lambda_{i+1}$. 
\end{enumerate}
To see (i), 
suppose there exists such  an $i$  that 
$\ell(\lambda)<i\le \ell(\eta)$ and $\eta_i\ge2$. Then
 $(\eta_2',2)\in\eta$, $\ell_\eta(\eta_2',2)=0$ and 
$a_\lambda(\eta_2',2)=-2$, in contradiction with \eqref{eq-eta-lambda-1}. 
Since $\eta_i>0$ for $i\le \ell(\eta)$, we must have $\eta_i=1$.

To see (ii), suppose there exists such an $i\le \min(\ell(\lambda),\ell(\eta))$ 
that
$\eta_i>\lambda_i+1$. Take the largest such $i$ and let $j=\lambda_i+2$. 
Then $(i,j)\in \eta$, $\ell_\eta(i,j)=0$ and $a_\lambda(i,j)=-2$, 
which contradicts \eqref{eq-eta-lambda-1}.

Under the assumption of (iii), 
we have $(i,\eta_i)\in\lambda$ and $a_\eta(i,\eta_i)=0$.
Further if $i=\ell(\lambda)$, or else $i<\ell(\lambda)$ 
and $\lambda_{i+1}<\eta_i$, then
we have $\ell_\lambda(i,\eta_i)=0$, which contradicts \eqref{eq-eta-lambda-2}.

To show (iv), first assume $i\le \ell(\lambda)-2$ and 
$\eta_i<\lambda_{i+1}$, $\eta_{i+1}=\lambda_{i+2}$ hold. Then 
$(i,\eta_i+1)\in\lambda$, $a_\eta(i,\eta_i+1)=-1$, and 
$\ell_{\lambda}(i,\eta_i+1)=1$, which contradicts 
\eqref{eq-eta-lambda-2}.
Similarly if $i=\ell(\lambda)-1$ and 
$\eta_{\ell(\lambda)-1}<\lambda_{\ell(\lambda)}$, then 
the same conclusion takes place.

Now we prove the Lemma. 
If $\eta_i=\lambda_i+1$ for all $i\le \ell(\lambda)$, then 
together with (i) we have $\eta=r_{\ell(\eta)}(\lambda)$.
Otherwise, from (ii), (iii) and (iv)
there exists an $n$ such that 
 $\eta_i=\lambda_i+1$ ($1\le i\le n$) and 
$\eta_i=\lambda_{i+1}$ for all $n<i\le \ell(\lambda)-1$.
Moreover if $\ell(\eta)\ge\ell(\lambda)$ then 
$\eta_{\ell(\lambda)}\le \eta_{\ell(\lambda)-1}=\lambda_{\ell(\lambda)}$
which contradicts (iii). Therefore $\ell(\eta)=\ell(\lambda)-1$ and
$\eta=r_n(\lambda)$.
\end{proof}

\begin{lem}\label{lem-N-factors}
Let $\lambda,\mu \in \Lambda$ and $n\in\Z_{\ge0}$. 
Using the notation \eqref{lambda-bar} and \eqref{rn}, we set 
\begin{align*}
&\ell=\ell(\lambda)\,,
\quad k=\ell(\mu)\,,
\quad \eta=r_n(\lambda)\,, \quad \gamma=(r_n(\mu'))'\,,
\\ 
&\tilde{\eta}=
\begin{cases}
\bar{\eta} & (n\le \ell-1),\\
(\lambda,1^{n-\ell+1}) & (n\ge \ell).
\end{cases}
\end{align*}
Then we have  
\begin{align}
\frac{N_{\eta,\lambda}(q^{-1})}{N_{\eta,\eta}(1)}=
&
\frac{N_{\tilde{\eta},\bar{\lambda}}(q^{-1})}
{N_{\tilde{\eta},\tilde{\eta}}(1)}
(1-q^{|\tilde{\eta}|-|\lambda|})\,,
\label{eq-N-q-1}
\\
N_{\mu,\eta}(u)=
&
N_{\mu,\tilde{\eta}}(q^{-1}u)
\prod_{j=1}^{\mu_\ell}
\frac{1-q^{j-1}u}{1-q^{-\ell_\mu(\ell,j)+j-2}u}
\prod_{i=1}^{\ell-1}(1-q^{\ell-i+a_\mu(i,1)}u)\,,
\label{eq-N-q-5}
\\
N_{\mu,\lambda}(u)
=&
N_{\mu,\bar{\lambda}}(q^{-1}u)
\prod_{j=1}^{\mu_{\ell+1}}
\frac{1-q^{j-1}u}{1-q^{-\ell_\mu(\ell+1,j)+j-2}u}
\prod_{i=1}^\ell(1-q^{\ell-i+a_\mu(i,1)+1}u)\,,
\label{eq-N-q-4}
\\
\frac{N_{\mu,\lambda}(u)}{N_{\mu,\eta}(qu)}
=&
\frac{N_{\mu,\bar{\lambda}}(q^{-1}u)}{N_{\mu,\tilde{\eta}}(u)}(1-u)\,, 
\label{eq-N-q-2}
\\
\frac{N_{\mu,\lambda}(u)}{N_{\mu,\eta}(qu)}
=&
\frac{N_{\bar{\mu},\lambda}(qu)}{N_{\bar{\mu},\eta}(q^2u)}
\frac{1-q^{|\eta|-|\lambda|+1-k}u}{1-qu}\,,
\label{eq-N-q-3}
\\
N_{\gamma,\lambda}(u)=
&N_{\gamma,\bar{\lambda}}(q^{-1}u)
(1-q^{\ell+|\gamma|-|\mu|-1}u)
\prod_{j=1}^{\mu_\ell}
\frac{1-q^{j-1}u}{1-q^{-\ell_\mu(\ell,j)+j-2}u}
\prod_{i=1}^{\ell-1}(1-q^{\ell-i+a_\mu(i,1)}u)\,.
\label{eq-N-q-6}
\end{align}
\qed
\end{lem}
\begin{proof}

By the definition we have 
\begin{align*}
a_{\bar{\lambda}}(i,j)=&a_\lambda(i,j+1)\quad (i=1,\ldots,\ell, \ j\in\Z_{>0}),
\\
\ell_{\bar{\lambda}}(i,j)=&\ell_\lambda(i,j+1)\quad ((i,j)\in \bar{\lambda}),
\\
a_{\tilde{\eta}}(i,j)=&a_\eta(i,j+1)\quad (i=1,\ldots, \min(\ell(\tilde{\eta}),\ell),\ j\in\Z_{>0}),
\\
\ell_{\tilde{\eta}}(i,j)=&
\begin{cases}
\ell_\eta(i,j+1) & ((i,j)\in \tilde{\eta},\ \text{$n\le \ell-1 $ or $n\ge \ell$ and $j\ge2$}),\\
n+1-i & (i>0, j=1, n\ge \ell).\\
\end{cases}
\end{align*}
Using these relations 
we show below \eqref{eq-N-q-1}--\eqref{eq-N-q-6}. 

\bigskip
 
\noindent{\underline{Proof of \eqref{eq-N-q-1}}}.\quad 
First, we consider 
 the case  $n\le \ell-1$. We have  
  \begin{align*}
  \frac{N_{\eta,\lambda}(q^{-1})}{N_{\tilde{\eta},\bar{\lambda}}(q^{-1})}=&
  \frac{\prod_{(i,j)\in\eta}\left(1-q^{-\ell_\eta(i,j)-a_\lambda(i,j)-2}\right)}
  {\prod_{(i,j)\in\tilde{\eta}}\left(1-q^{-\ell_\eta(i,j+1)-a_\lambda(i,j+1)-2}\right)}
  \frac{\prod_{(i,j)\in\lambda}\left(1-q^{\ell_\lambda(i,j)+a_\eta(i,j)}\right)}
  {\prod_{(i,j)\in\bar{\lambda}, i<\ell}\left(1-q^{\ell_\lambda(i,j+1)+a_\eta(i,j+1)}\right)}
\\
&\times 
  \frac{1}{\prod_{j=1}^{\lambda_\ell-1}\left(1-q^{\ell_\lambda(\ell,j+1)-j}\right)}
  \\
=&
\prod_{i=1}^{\ell}\left(1-q^{-\ell+i-\lambda_i}\right)
\prod_{i=1}^{\ell-1}\left(1-q^{\ell-i+\eta_i-1}\right)\,,
  \end{align*}
  and 
 \begin{align*}
 \frac{N_{\tilde{\eta},\tilde{\eta}}(1)}{N_{\eta,\eta}(1)}=&
  \frac{\prod_{(i,j)\in\tilde{\eta}}\left(1-q^{-\ell_\eta(i,j+1)-a_\eta(i,j+1)-1}\right)}
  {\prod_{(i,j)\in\eta}\left(1-q^{-\ell_\eta(i,j)-
a_\eta(i,j)   
-1}\right)}
 \frac{\prod_{(i,j)\in\tilde{\eta}}\left(1-q^{\ell_\eta(i,j+1)+a_\eta(i,j+1)+1}\right)}
  {\prod_{(i,j)\in\eta}\left(1-q^{\ell_\eta(i,j)+
a_\eta(i,j)  
+1}\right)}
  \\
  =&
\prod_{i=1}^{\ell-1}
\frac{1}{\left(1-q^{-\ell+i+1-\eta_i}\right)
\left(1-q^{\ell-i+\eta_i-1}\right)} \,.
 \end{align*}
Since $\eta=(\lambda_1+1,\ldots,\lambda_n+1,\lambda_{n+2},\ldots,\lambda_\ell)$
and  $|\tilde{\eta}|-|\lambda|=-\ell+ n+1-\lambda_{n+1}$, we obtain \eqref{eq-N-q-1} for $n\le \ell-1$.

Second, we consider the case    $n\ge \ell$. 
 We have      
 \begin{align*}
 \frac{N_{\eta,\lambda}(q^{-1})}{N_{\tilde{\eta},\bar{\lambda}}(q^{-1})}=
&
\frac{\prod_{i=1}^{n}\left( 1-q^{-(n-i)-(\lambda_i-1)-2}\right)
\prod_{i=1}^{\ell}\left( 1-q^{-(\ell-i)-(\lambda_i-2)-2}\right)
\left( 1-q^{\ell-i+\lambda_i}\right)} 
{\prod_{i=1}^{\ell}\left( 1-q^{-(n+1-i)-(\lambda_i-2)-2}\right)
\prod_{i=\ell+1}^{n+1}\left( 1-q^{-(n+1-i)-(\lambda_i-1)-2}\right)
}
\\
=&\frac{\prod_{i=1}^\ell\left(1-q^{-\ell+i-\lambda_i}\right)
\left(1-q^{\ell-i+\lambda_i}\right)}  
{(1-q^{-(n-\ell)-1})}, 
 \end{align*}
 and 
 \begin{align*}
 \frac{N_{\tilde{\eta},\tilde{\eta}}(1)}{N_{\eta,\eta}(1)}
 =&\frac{(1-q^{-(n-\ell)-1})(1-q^{n-\ell+1})}
 {\prod_{i=1}^\ell\left(1-q^{-\ell+i-\lambda_i}\right)
\left(1-q^{\ell-i+\lambda_i}\right)  
}. 
 \end{align*}
 Since $|\tilde{\eta}|-|\lambda|=n+1-\ell$, we obtain \eqref{eq-N-q-1} for $n\ge \ell$.
\bigskip

 \noindent{\underline{Proof of \eqref{eq-N-q-5} }}.\quad  
First, we consider the case 
 $n\le \ell-1$. 
 Then by $\ell(\tilde{\eta})\le \ell-1$,   
we obtain 
\begin{align*}
 \frac{N_{\mu,\eta}(u)}{N_{\mu,\tilde{\eta}}(q^{-1}u)}=&
 \prod_{(i,j)\in\mu, i\ge \ell}
\frac{1-q^{-\ell_\mu(i,j)+j-1}u}{1-q^{-\ell_\mu(i,j)+j-2}u}
\prod_{i=1}^{\ell-1}(1-q^{\ell-1-i+a_\mu(i,1)+1}u)
\\
=&\prod_{j=1}^{\mu_\ell}
\frac{1-q^{j-1}u}{1-q^{-\ell_\mu(\ell,j)+j-2}u}
\prod_{i=1}^{\ell-1}(1-q^{\ell-i+a_\mu(i,1)}u). 
 \end{align*} 
 Second, we consider the case 
  $n\ge \ell$.   
 Then, we have $a_{\tilde{\eta}}(i+1,j)=a_\eta(i,j)$ if $\ell<i\le n$ and 
 $\ell_{\tilde{\eta}}(i,1)=\ell_\eta(i,1)+1$ for $i\in\Z_{>0}$. Hence, we obtain
 \begin{align*}
 \frac{N_{\mu,\eta}(u)}{N_{\mu,\tilde{\eta}}(q^{-1}u)}
 =&\frac{\prod_{i=\ell+1}^{\ell(\mu)}
 \prod_{j=\mu_{i+1}+1}^{\mu_i}(1-q^{-a_\eta(i,j)-1}u)}
 {\prod_{j=1}^{\mu_{\ell+1}}(1-q^{-\ell_\mu(\ell+1,j)+j-3}u)}
 \frac{\prod_{i=1}^\ell(1-q^{\ell-i+a_\mu(i,1)}u)}
 {(1-q^{a_\mu(n+1,1)}u)}
 \\
 =&\frac{\prod_{j=\mu_{n+1}+1}^{\mu_\ell}(1-q^{j-2}u)\prod_{j=1}^{\mu_{n+1}}(1-q^{j-1}u)}
 {\prod_{j=1}^{\mu_{\ell}}(1-q^{-\ell_\mu(\ell+1,j)+j-3}u)}
 \frac{\prod_{i=1}^{\ell-1}(1-q^{\ell-i+a_\mu(i,1)}u)}
 {(1-q^{\mu_{n+1}-1}u)}(1-q^{\mu_{\ell}-1}u)
 \\
 =&\prod_{j=1}^{\mu_\ell}
\frac{1-q^{j-1}u}{1-q^{-\ell_\mu(\ell,j)+j-2}u}
\prod_{i=1}^{\ell-1}(1-q^{\ell-i+a_\mu(i,1)}u). 
 \end{align*}
 \bigskip

 \noindent{\underline{Proof of \eqref{eq-N-q-4} }}.\quad  
In \eqref{eq-N-q-5}, choose $n=\ell-1$, then rename 
$\ell$ by $\ell+1$ and $\lambda_i$ by $\lambda_i-1$. 
We obtain \eqref{eq-N-q-4}.
 \bigskip

 \noindent{\underline{Proof of \eqref{eq-N-q-2} }}.\quad  
This follows from \eqref{eq-N-q-4} and \eqref{eq-N-q-5}.
 \bigskip

\noindent{\underline{Proof of \eqref{eq-N-q-3} }}.\quad  
Using the relation  \eqref{rule1} 
and \eqref{eq-N-q-4}, we have 
\begin{align*}
\frac{N_{\mu,\lambda}(u)N_{\bar{\mu},\eta}(q^2u)}
{N_{\bar{\mu},\lambda}(qu)N_{\mu,\eta}(qu)}=&
q^{|\eta|-|\lambda|-k
}
\frac{N_{\lambda,\mu}(u^{-1})N_{\eta,\bar{\mu}}((q^2u)^{-1})}
{N_{\lambda,\bar{\mu}}((qu)^{-1})N_{\eta,\mu}((qu)^{-1})}
\\
=&
q^{|\eta|-|\lambda|-k}
\prod_{j=1}^{\lambda_{k+1}}\frac{1-q^{j-1}u^{-1}}
{1-q^{-\ell_\lambda(k+1,j)+j-2}u^{-1}}
\prod_{j=1}^{\eta_{k+1}}\frac{1-q^{-\ell_\eta(k+1,j)+j-3}u^{-1}}
{1-q^{j-2}u^{-1}}
\\
&\times 
\prod_{i=1}^k\frac{1-q^{k-i+a_\lambda(i,1)+1}u^{-1}}
{1-q^{k-i+a_\eta(i,1)}u^{-1}}. 
\end{align*}
First, we consider the case 
 $n\le \ell-1$. 
By definition, we have for $k+1\le n$ 
 \begin{align*}
 \ell_\eta(k+1,j)=\begin{cases}
 \ell_\lambda(k+1,j)-1& (j=1,\ldots,\lambda_{n+1}),
\\
n-(k+1)& (j=\lambda_{n+1}+1), 
 \\
 \ell_\lambda(k+1,j-1)& (j=\lambda_{n+1}+2,\ldots,
\lambda_1+1)\,,
\end{cases}
 \end{align*}
and for  $k+1\ge n+1$
\begin{align*}
\ell_\eta(k+1,j)=\ell_\lambda(k+1,j)-1\quad (j=1,\ldots,\eta_{k+1}). 
\end{align*}
 Hence, we obtain for $k+1\le n$, 
 \begin{align*}
 \frac{N_{\mu,\lambda}(u)N_{\bar{\mu},\eta}(q^2u)}
{N_{\bar{\mu},\lambda}(qu)N_{\mu,\eta}(qu)}=&
q^{|\eta|-|\lambda|-k  
}
\prod_{j=1}^{\lambda_{k+1}}\frac{1-q^{j-1}u^{-1}}
{1-q^{-\ell_\lambda(k+1,j)+j-2}u^{-1}}
\prod_{j=1}^{\lambda_{n+1}}\left(1-q^{-\ell_\lambda(k+1,j)+j-2}u^{-1}\right)
\\
&\times\left(1-q^{k+1-n+\lambda_{n+1}-2}u^{-1}\right)\frac
{
\prod_{j=\lambda_{n+1}+2}^{\lambda_{k+1}+1}
\left(1-q^{-\ell_\lambda(k+1,j-1)+j-3}u^{-1}\right)}
{\prod_{j=1}^{\lambda_{k+1}+1}\left(1-q^{j-2}u^{-1}\right)}
\\
=&q^{|\eta|-|\lambda|-k  
}\times
\frac{1-q^{k-n+\lambda_{n+1}-1}u^{-1}}
{1-q^{-1}u^{-1}}. 
 \end{align*}
We also obtain for $k+1\ge n+1$, 
\begin{align*}
\frac{N_{\mu,\lambda}(u)N_{\bar{\mu},\eta}(q^2u)}
{N_{\bar{\mu},\lambda}(qu)N_{\mu,\eta}(qu)}=
&q^{|\eta|-|\lambda|-k  
}
\prod_{j=1}^{\lambda_{k+1}}\frac{1-q^{j-1}u^{-1}}
{1-q^{-\ell_\lambda(k+1,j)+j-2}u^{-1}}
\prod_{j=1}^{\lambda_{k+2}}
\frac{1-q^{-\ell_\lambda(k+1,j)+j-2}u^{-1}}
{1-q^{j-2}u^{-1}}
\\
&\times 
\prod_{i=n+1}^k\frac{1-q^{k-i+a_\lambda(i,1)+1}u^{-1}}
{1-q^{k-i+a_\lambda(i+1,1)}u^{-1}}
\\
=&q^{|\eta|-|\lambda|-k
}\times
\frac{1-q^{k-n+\lambda_{n+1}-1}u^{-1}}
{1-q^{-1}u^{-1}}. 
\end{align*} 
 Since $|\eta|-|\lambda|=n-\lambda_{n+1}$, 
we obtain \eqref{eq-N-q-3} for $n\le \ell-1$. 
 
 Second, we consider the case 
  $n\ge \ell$.  By the definition, we have 
\begin{align*}
\ell_\eta(i,j)=&\ell_\lambda(i,j-1)\quad ((i,j)\in\eta,\ j\ge 2),
\\
\ell_\eta(i,1)=&n-i\quad (i\in\Z_{>0}).  
\end{align*}
 Hence, we obtain 
 \begin{align*}
 &\frac{N_{\mu,\lambda}(u)N_{\bar{\mu},\eta}(q^2u)}
{N_{\bar{\mu},\lambda}(qu)N_{\mu,\eta}(qu)}
\\
&=
q^{|\eta|-|\lambda|-k  
}
\prod_{j=1}^{\lambda_{k+1}}
\frac{1-q^{j-1}u^{-1}}
{1-q^{-\ell_\lambda(k+1,j)+j-2}u^{-1}}
\prod_{j=2}^{\eta_{k+1}}
\frac{1-q^{-\ell_\lambda(k+1,j-1)+j-3}u^{-1}}
{1-q^{j-2}u^{-1}}\cdot
\frac{1-q^{n+k-1}u^{-1}}{1-q^{-1}u^{-1}}
\\
&\times 
\prod_{i=n+1}^k\frac{1-q^{k-i}u^{-1}}
{1-q^{k-i-1}u^{-1}}
\\
&=q^{|\eta|-|\lambda|-k  
}\times \frac{1-q^{-n+k-1}u^{-1}}{1-q^{-1}u^{-1}}. 
 \end{align*}
 Since $|\eta|-|\lambda|=n$, 
we obtain \eqref{eq-N-q-3} for $n\ge \ell$. 
 \bigskip
 
 \noindent{\underline{Proof of \eqref{eq-N-q-6}}}.\quad  
By \eqref{eq-N-q-4}, we have 
\begin{align*}
N_{\gamma,\lambda}(u)=
&N_{\gamma,\bar{\lambda}}(q^{-1}u)
\prod_{j=1}^{\gamma_{\ell+1}}
\frac{1-q^{j-1}u}{1-q^{-\ell_\gamma(\ell+1,j)+j-2}u}\,
\prod_{i=1}^{\ell}(1-q^{l-i+a_\gamma(i,1)+1}u).
\end{align*}
Recall that $\gamma=(s_n(\mu'))'$. 
First, we consider the case 
  $n\le \mu_1-1$. By the definition, we have 
  \begin{align*}
  \ell_\gamma(i,j)=&\begin{cases}
 \ell_\mu(i,j)+1&((i,j)\in\gamma, j\le n),
 \\
 \ell_\mu(i,j+1)& ((i,j)\in\gamma, j\ge n+1),
 \end{cases}
\\
a_\gamma(i,j)=& \begin{cases}
a_\mu(i,j)-1&(i=1,\ldots, \mu'_{n+1}, j\in\Z_{>0}),
\\
n-j&(i=\mu'_{n+1}+1, j\in\Z_{>0}),
\\
a_\mu(i-1,j)&((i,j)\in\gamma, i\ge \mu'_{n+1}+2, j\in\Z_{>0}). 
\end{cases}
  \end{align*}
Hence, for $\ell+1\le \mu'_{n+1}$ we have
\begin{align*}
N_{\gamma,\lambda}(u)=
&N_{\gamma,\bar{\lambda}}(q^{-1}u)
\frac{\prod_{j=1}^{\mu_{\ell+1}-1}\left(1-q^{j-1}u\right)}
{\prod_{j=1}^n\left(1-q^{-\ell_\mu(\ell+1,j)+j-3}u\right)
\prod_{j=n+1}^{\mu_{\ell+1}}\left(1-q^{-\ell_\mu(\ell+1,j+1)+j-2}u\right)}
\\
&\times 
\prod_{i=1}^{\ell}(1-q^{l-i+a_\mu(i,1)}u)
\\
=&N_{\gamma,\bar{\lambda}}(q^{-1}u)
\prod_{j=1}^{\mu_\ell}
\frac{1-q^{j-1}u}{1-q^{-\ell_\mu(\ell,j)+j-2}u}
\cdot \frac{1-q^{-\mu'_{n+1}+\ell+n-1}u}
{1-q^{\mu_\ell-1}u}
\\
&\times 
\prod_{i=1}^{\ell-1}(1-q^{l-i+a_\mu(i,1)}u)
\cdot \left(1-q^{\mu_\ell-1}u\right). 
\end{align*}

Similarly, for 
$\ell+1=\mu'_{n+1}+1$ 
we have 
\begin{align*}
N_{\gamma,\lambda}(u)=
&N_{\gamma,\bar{\lambda}}(q^{-1}u)
\prod_{j=1}^n
\frac{1-q^{j-1}u}{1-q^{-\ell_\mu(\ell+1,j)+j-3}u}
\prod_{i=1}^{\ell}(1-q^{l-i+a_\mu(i,1)}u)\cdot\left(1-q^{\mu_\ell-1}u\right)
\\
=&N_{\gamma,\bar{\lambda}}(q^{-1}u)
\prod_{j=1}^{\mu_\ell}
\frac{1-q^{j-1}u}{1-q^{-\ell_\mu(\ell,j)+j-2}u}
\cdot \frac{1-q^{n-1}u}
{1-q^{\mu_\ell-1}u}
\prod_{i=1}^{\ell-1}(1-q^{l-i+a_\mu(i,1)}u)\cdot\left(1-q^{\mu_\ell-1}u\right). 
\end{align*}
Here, we use the inequality $\mu_{\ell+1}\le n$. 

For $\ell+1\ge \mu'_{n+1}+2$ we have 
\begin{align*}
N_{\gamma,\lambda}(u)=
&N_{\gamma,\bar{\lambda}}(q^{-1}u)
\prod_{j=1}^{\mu_\ell}
\frac{1-q^{j-1}u}{1-q^{-\ell_\mu(\ell,j)+j-2}u}
\prod_{i=1}^{\mu'_{n+1}}\left(1-q^{\ell-i+a_\mu(i,1)}u\right)
\cdot \left(1-q^{\ell-\mu'_{n+1}+n-1}u\right)
\\
&\times 
\prod_{i=\mu'_{n+1}+2}^\ell\left(1-q^{\ell-(i-1)+a_\mu(i-1,1)}u\right)
\\
=&N_{\gamma,\bar{\lambda}}(q^{-1}u)
\prod_{j=1}^{\mu_\ell}
\frac{1-q^{j-1}u}{1-q^{-\ell_\mu(\ell,j)+j-2}u}
\prod_{i=1}^{\ell-1}(1-q^{l-i+a_\mu(i,1)}u)\cdot 
\left(1-q^{\ell-\mu'_{n+1}+n-1}u\right). 
\end{align*}
Here, we use the inequality $\mu_{\ell}\le n$. 
Since $|\gamma|-|\mu|=n-\mu'_{n+1}$, we obtain \eqref{eq-N-q-6} for $n\le \mu_1-1$. 

Second, we consider the case 
  $n\ge \mu_1$. By the definition, we have 
  \begin{align*}
 \ell_\gamma(i,j)=&
 \ell_\mu(i-1,j)\quad  ((i,j)\in\gamma, i\ge 2),
\\
a_\gamma(i,j)=& \begin{cases}
n-j&(i=1, j\in\Z_{>0}),
\\
a_\mu(i-1,j)&((i,j)\in\gamma, i\ge 2, j\in\Z_{>0}). 
\end{cases}
  \end{align*}
  Hence, we obtain
  \begin{align*}
  N_{\gamma,\lambda}(u)=
&N_{\gamma,\bar{\lambda}}(q^{-1}u)
\prod_{j=1}^{\mu_\ell}
\frac{1-q^{j-1}u}
{1-q^{-\ell_\mu(\ell,j)+j-2}u}
\prod_{i=2}^\ell\left(1-q^{\ell-i+a_\mu(i-1,1)+1}u\right)\cdot \left(1-q^{\ell-1+n}u\right). 
  \end{align*}
  Since $|\gamma|-|\mu|=n$, we obtain \eqref{eq-N-q-6} for $n\ge \mu_1$. 
\end{proof}

Lemma \ref{lem-N-factors} enables us to reduce $\lambda$ to $\bar\lambda$ 
of $X_{\lambda, \mu, \alpha, \beta}
^{\e}$ and $Y_{\lambda, \mu, \alpha, \beta}
^{\e}$. Denote by $T_1$  
 the $q$-shift operator with respect to $x_1$:
$(T_1f)(x_1,\ldots,x_n)=f(qx_1,\ldots,x_n)$. 
\begin{lem}\label{lem-lambda}
For partitions $\lambda$, $\mu$, $\alpha$, $\beta$, we have 
\begin{align*}
X_{\lambda, \mu, \alpha, \beta}
^{\e}(\theta_\infty,\theta_1,\sigma;x_1,x_2)=&
C(1-q^{-\e\theta_\infty-\theta_1-\sigma-1/2})q^{\e\theta_\infty-\theta_1-\sigma-1/2}
\\
&\times (q^{-\ell(\lambda) 
-\e\theta_\infty+\theta_1+\sigma+1/2}T_1-1)
\left( X_{\bar\lambda 
, \mu, \alpha, \beta}
^{+}(\theta_\infty,\theta_1+\half,\sigma+\half;x_1,x_2)\right),
\\
Y_{\lambda, \mu, \alpha, \beta}
^{+}(\theta_\infty,\theta_1,\sigma;x_1,x_2)=&
C\frac{(1-q^{\pm\theta_\infty-\theta_1-\sigma-1/2})}{(1-q^{2\sigma+1})}
\\
&\times(1-q^{-\ell(\lambda)+1+2\sigma}T_1)\left(
Y_{\bar\lambda 
, \mu, \alpha, \beta}
^{+}(\theta_\infty,\theta_1+\half,\sigma+\half;x_1,x_2)
\right),
\\
Y_{\lambda, \mu, \alpha, \beta}
^{-}(\theta_\infty,\theta_1,\sigma;x_1,x_2)=&C(1-q^{-2\sigma})
\frac{x_2}{x_1}
\\
&\times (1-q^{-\ell(\lambda)}T_1)
\left(
Y_{\bar\lambda 
, \mu, \alpha, \beta}
^{-}(\theta_\infty,\theta_1+\half,\sigma+\half;x_1,x_2)
\right),
\end{align*}
where 
\begin{align*}
C=&\prod_{j=1}^{\alpha_{\ell(\lambda)}}
\frac{1-q^{j+\theta_\infty-\theta_1-\sigma-1/2}}
{1-q^{-\ell_\alpha(\ell(\lambda),j)+j+\theta_\infty-\theta_1-\sigma-3/2}}
\prod_{i=1}^{\ell(\lambda)-1}(1-q^{\ell(\lambda)-i+a_\alpha(i,1)+\theta_\infty-\theta_1-\sigma+1/2})
\\
&\times\prod_{j=1}^{\beta_{\ell(\lambda)}}
\frac{1-q^{j-\theta_\infty-\theta_1-\sigma-1/2}}
{1-q^{-\ell_\beta(\ell(\lambda),j)+j-\theta_\infty-\theta_1-\sigma-3/2}}
\prod_{i=1}^{\ell(\lambda)-1}(1-q^{\ell(\lambda)-i+a_\beta(i,1)-\theta_\infty-\theta_1-\sigma+1/2}) 
\\
&\times q^{\ell(\lambda)-|\alpha|-|\beta|}x_2^{-\ell(\lambda)}. 
\end{align*}
\qed
\end{lem}
\bigskip

\textbf{Proof of Proposition \ref{prop-braiding-relation}.}
By induction. 
If all partitions $\lambda$, $\mu$, $\alpha$ and $\beta$ are equal to $\emptyset$, 
then the identity \eqref{eq-general-partitions} 
is the connection formula of basic hypergeometric functions. 
By definition we have  the symmetries  
\begin{align*}
&X_{\lambda, \mu, \alpha, \beta}
^{\e}(\theta_\infty,\theta_1,\sigma;x_1,x_2)=
X_{\mu,\lambda, \alpha, \beta}
^{\e}(\theta_\infty,\theta_1,-\sigma;x_1,x_2),
\\
&X_{\lambda, \mu, \alpha, \beta}
^{-}(\theta_\infty,\theta_1,\sigma;x_1,x_2)=
X_{\lambda, \mu,  \beta,\alpha}
^{+}(-\theta_\infty,\theta_1,\sigma;x_1,x_2),
\\
&Y_{\lambda, \mu, \alpha, \beta}
^{-}(\theta_\infty,\theta_1,\sigma;x_1,x_2)=
Y_{ \mu,\lambda, \alpha, \beta}
^{+}(\theta_\infty,\theta_1,-\sigma;x_1,x_2),
\\
&Y_{\lambda, \mu, \alpha, \beta}
^{\e}(\theta_\infty,\theta_1,\sigma;x_1,x_2)=
Y_{\lambda, \mu,  \beta,\alpha}
^{\e}(-\theta_\infty,\theta_1,\sigma;x_1,x_2),
\\
&\mathcal{B}_{-,\e}\left[\begin{matrix}
\theta_1&\half
\\
\theta_\infty&\sigma
\end{matrix} \Bigl| \frac{x_2}{x_1}\right]=\mathcal{B}_{+,\e}\left[\begin{matrix}
\theta_1&\half
\\
\theta_\infty&-\sigma
\end{matrix} \Bigl| \frac{x_2}{x_1}\right],\quad 
\B_{\e,-}\left[\begin{matrix}
\theta_1&\half
\\
\theta_\infty&\sigma
\end{matrix} \Bigl| \frac{x_2}{x_1}\right]=\B_{\e,+}\left[\begin{matrix}
\theta_1&\half
\\
-\theta_\infty&\sigma
\end{matrix} \Bigl| \frac{x_2}{x_1}\right],
\end{align*}
and by Lemma \ref{lem-Nfactor-symmetry} we have 
\begin{equation*}
X_{\lambda, \mu, \alpha, \beta}
^{\e}(\theta_\infty,\theta_1,\sigma;x_1,x_2)=
Y_{\beta',\alpha',\mu',\lambda'}
^{\e}(\sigma,\theta_1,\theta_\infty;(qx_1)^{-1},(q^{2\theta_1}x_2)^{-1}).
\end{equation*}
Recall also the relations \eqref{braid-q-per} and \eqref{braid-inv} 
for
the braiding matrix. 
Because of these relations, it suffices to show that 
for any partitions 
$\lambda$, $\mu$, $\alpha$ and $\beta$, 
 the connection formula \eqref{eq-general-partitions} is equivalent to the formula \eqref{eq-general-partitions} for 
 $\bar\lambda$,
$\mu$, $\alpha$ and $\beta$ 
acting by a $q$-difference operator.  
 
By Lemma \ref{lem-lambda}, the left hand side of \eqref{eq-general-partitions} is computed as follows:
\begin{align*}
&q^{\theta_\infty^2}\left(\frac{q^{2\theta_1} x_2}{x_1}\right)^{\theta_\infty+1/4}
\frac{
\prod_{\e''=\pm}
\Gamma_q(\half+\theta_\infty-\theta_1+\e'' \sigma)
}{\Gamma_q(1+2\theta_\infty)}X_{\lambda,\mu,\alpha,\beta}
^{+}(\theta_\infty,\theta_1,\sigma;x_1,x_2)
\\
&=q^{\theta_\infty^2}\left(\frac{q^{2\theta_1+1} x_2}{x_1}\right)^{\theta_\infty+1/4}
q^{-\theta_\infty-1/4}
\frac{
\prod_{\e''=\pm}
\Gamma_q(\half+\theta_\infty-(\theta_1+\half)+\e'' (\sigma+\half))
}
{\Gamma_q(1+2\theta_\infty)}
\frac{1-q^{\pm\theta_\infty-\theta_1-\sigma-1/2}}{1-q}
\\
&\times Cq^{\theta_\infty-\theta_1-\sigma-1/2}
 (q^{-\ell(\lambda)
 -\theta_\infty+\theta_1+\sigma+1/2}T_1-1)
\left( X_{\bar\lambda
, \mu, \alpha, \beta}
^{+}(\theta_\infty,\theta_1+\half,\sigma+\half;x_1,x_2)\right)
\\
&=Cq^{-\theta_1-\sigma-3/4}
\frac{1-q^{\pm\theta_\infty-\theta_1-\sigma-1/2}}{1-q}
q^{\theta_\infty^2}
\frac{
\prod_{\e''=\pm}
\Gamma_q(\half+\theta_\infty-(\theta_1+\half)+\e'' (\sigma+\half))
}
{\Gamma_q(1+2\theta_\infty)}
\\
&\times 
\left( q^{-\ell(\lambda)
+\theta_1+\sigma+3/4}T_1-1\right)\left( 
\left(\frac{q^{2\theta_1+1} x_2}{x_1}\right)^{\theta_\infty+1/4}
X_{\bar\lambda 
, \mu, \alpha, \beta}
^{+}(\theta_\infty-\half,\theta_1+\half,\sigma)
\right). 
\end{align*}
Similarly, by Lemma \ref{lem-lambda}, the right hand side of \eqref{eq-general-partitions} 
is computed as follows:
\begin{align*}
&\sum_{\e=\pm}q^{\sigma^2}\left(\frac{q x_1}{x_2}\right)^{\e\sigma+1/4}
\frac{
\prod_{\e''=\pm}
\Gamma_q(\half+\e''\theta_\infty-\theta_1+\e \sigma)
}{\Gamma_q(1+2\e\sigma)}
 Y_{\lambda, \mu, \alpha, \beta}
^{\e}(\theta_\infty,\theta_1,\sigma;x_1,x_2)
\\
&\times 
 \B_{\e,+}\left[\begin{matrix}
\theta_1&\half 
\\
\theta_\infty&\sigma
\end{matrix} \Bigl| \frac{x_2}{x_1}\right]
q^{\theta_1^2-\theta_1/2}\left(\frac{x_2}{x_1}\right)^{\theta_1}
\\
&=\sum_{\e=\pm}q^{(\sigma+1/2)^2}\left(\frac{q x_1}{x_2}\right)^{\e(\sigma+1/2)+1/4}
\frac{
\prod_{\e''=\pm}
\Gamma_q(\half+\e''\theta_\infty-(\theta_1+\half)+\e (\sigma+\half))
}
{\Gamma_q(1+2\e(\sigma+\half))}
\frac{1-q^{\pm\theta_\infty-\theta_1-\sigma-1/2}}{1-q}
\\
 &\times \B_{\e,+}\left[\begin{matrix}
\theta_1+\half&\half
\\
\theta_\infty&\sigma+\half
\end{matrix} \Bigl| \frac{x_2}{x_1}\right]
q^{(\theta_1+1/2)^2-(\theta_1+1/2)/2}\left(\frac{x_2}{x_1}\right)^{\theta_1+1/2}
q^{-\theta_1-\sigma-3/4}
\\
&\times
C(q^{-\ell(\lambda)+(1+\e)(\sigma+1/2)}T_1-1)
    \left( Y_{\bar\lambda 
    , \mu, \alpha, \beta}^{\e}
(\theta_\infty,\theta_1+\half,\sigma+\half)\right)
\\
&=Cq^{-\theta_1-\sigma-3/4}
\frac{1-q^{\pm\theta_\infty-\theta_1-\sigma-1/2}}{1-q}
\sum_{\e=\pm}q^{(\sigma+1/2)^2}
\frac{
\prod_{\e''=\pm}
\Gamma_q(\half+\e''\theta_\infty-(\theta_1+\half)+\e (\sigma+\half))
}
{\Gamma_q(1+2\e(\sigma+\half))}
\\
&\times \left\{(q^{-\ell(\lambda)+\theta_1+\sigma+3/4}T_1-1)
    \left( \left(\frac{q x_1}{x_2}\right)^{\e(\sigma+1/2)+1/4}Y_{\bar\lambda
, \mu, \alpha, \beta}^{\e}
(\theta_\infty,\theta_1+\half,\sigma+\half)
 \right.\right.
 \\
 &\times \left.\left.\B_{\e,+}\left[\begin{matrix}
\theta_1+\half&\half
\\
\theta_\infty&\sigma+\half
\end{matrix} \Bigl| \frac{x_2}{x_1}\right]q^{(\theta_1+1/2)^2-(\theta_1+1/2)/2}\left(\frac{x_2}{x_1}\right)^{\theta_1+1/2}\right)\right\}. 
\end{align*}
Therefore, the identity \eqref{eq-general-partitions} with parameters $\theta_\infty,\theta_1,\sigma$ 
for $\lambda,\mu,\alpha,\beta$ is 
reduced to the identity \eqref{eq-general-partitions} with parameters $\theta_\infty,\theta_1+1/2,\sigma+1/2$ 
for $\bar\lambda 
,\mu,\alpha,\beta$. 
\qed

\vskip 1cm

 \textbf{Acknowledgments.}
The authors 
would like to thank 
Kazuki Hiroe, 
Oleg Lisovyy, 
Toshiyuki Mano, 
Yousuke Ohyama, 
Yusuke Ohkubo,
Teruhisa Tsuda, 
Yasuhiko Yamada and 
Shintaro Yanagida
 for discussions. 

MJ is partially supported by JSPS KAKENHI Grant Number JP16K05183. 
HN is partially supported by JSPS KAKENHI Grant Number JP15K17560.
HS is partially supported by JSPS KAKENHI Grant Number JP15K04894.

\end{document}